\documentclass[envcountsame]{llncs}

\usepackage{tikz}
\usetikzlibrary{shapes,arrows}
\usepackage{amsmath}
\usepackage{amssymb}
\usepackage{amstext}
\usepackage{todo}
\usepackage[vlined,boxed]{algorithm2e}
\usepackage{alltt}
\usepackage{wrapfig}
\usepackage{enumitem}

\tikzstyle{line} = [draw, color=black, -latex]

\newcommand{\leaveOut}[1]{}

\begin{document}

\title{The Polynomial Complexity of Vector Addition Systems with States}

\author{
Florian Zuleger\\
zuleger@forsyte.tuwien.ac.at
}
\institute{TU Wien}
\maketitle

%neue macros
\newcommand{\ctree}{\ensuremath{T}}
\newcommand{\vExp}{\ensuremath{\mathtt{vExp}}}
\newcommand{\tExp}{\ensuremath{\mathtt{tExp}}}
\newcommand{\nobound}{\ensuremath{\infty}}
\newcommand{\node}{\ensuremath{\eta}}
\newcommand{\nats}{\ensuremath{\mathbb{N}}}
\newcommand{\possibleBoundSet}{\ensuremath{\mathit{RelevantLayers}}}
\newcommand{\coeff}{\ensuremath{w}}
\newcommand{\instances}{\ensuremath{\mathtt{instance}}}
\newcommand{\tbound}{\ensuremath{\mathtt{tbound}}}
\newcommand{\vbound}{\ensuremath{\mathtt{vbound}}}
\newcommand{\configuration}{\ensuremath{\kappa}}
\newcommand{\indeX}{\ensuremath{k}}
\newcommand{\layer}{\ensuremath{l}}
\newcommand{\layerF}{\ensuremath{\mathtt{layer}}}
\newcommand{\ppath}{\ensuremath{\sigma}}
\newcommand{\ppathLayer}{\ensuremath{\tau}}
\newcommand{\ppathComplete}{\ensuremath{\rho}}
\newcommand{\wit}{\ensuremath{\alpha}}
\newcommand{\witLayer}{\ensuremath{\beta}}
\newcommand{\witComplete}{\ensuremath{\gamma}}
\newcommand{\cycleF}{\ensuremath{\mathtt{cycle}}}
\newcommand{\NoC}{\ensuremath{d}}
\newcommand{\boundDom}{\ensuremath{\mathtt{varDom}}}
\newcommand{\boundDim}{\ensuremath{\mathtt{varDim}}}
\newcommand{\pot}{\ensuremath{\mathtt{pot}}}
\newcommand{\inv}{\ensuremath{\mathtt{inv}}}
\newcommand{\Span}{\ensuremath{\mathtt{span}}}
\newcommand{\varsF}{\ensuremath{\mathtt{Vars}}}
\newcommand{\partitionF}{\ensuremath{\mathtt{node}}}
\newcommand{\unboundedTransitions}{\ensuremath{U}}
\newcommand{\boundedTransitions}{\ensuremath{B}}
\newcommand{\decreasingTransitions}{\ensuremath{R}}
\newcommand{\initBound}{\ensuremath{\mathtt{init}}}
\newcommand{\increase}{\ensuremath{\mathtt{inc}}}
\newcommand{\varsum}{\ensuremath{\mathtt{varsum}}}
\newcommand{\Root}{\ensuremath{\iota}}
\newcommand{\maxentry}{\ensuremath{a}}
\newcommand{\extended}{\ensuremath{\mathit{ext}}}
\newcommand{\extF}{\ensuremath{\mathtt{ext}}}
\newcommand{\length}{\ensuremath{\mathit{length}}}

%alte macros
\newcommand{\field}{\ensuremath{\mathbb{S}}}

\newcommand{\vass}{\ensuremath{\mathcal{V}}}
\newcommand{\vassAlt}{\ensuremath{\mathcal{W}}}
\newcommand{\vassF}{\ensuremath{\mathtt{VASS}}}
\newcommand{\vars}{\ensuremath{\mathit{Var}}}
\newcommand{\var}{\ensuremath{x}}
\newcommand{\state}{\ensuremath{s}}
\newcommand{\states}{\mathit{St}}
\newcommand{\transitions}{\mathit{Trns}}
\newcommand{\edges}{\ensuremath{E}}
\newcommand{\update}{\ensuremath{d}}
\newcommand{\updates}{\ensuremath{D}}
\newcommand{\paath}{\ensuremath{\pi}}
\newcommand{\trace}{\ensuremath{\zeta}}
\newcommand{\val}{\ensuremath{\nu}}
\newcommand{\configs}{\ensuremath{\mathit{Cfg}}}
\newcommand{\config}{\ensuremath{\sigma}}
\newcommand{\dimension}{\ensuremath{n}}
\newcommand{\dimOp}{\ensuremath{\mathit{dim}}}
\newcommand{\Val}{\ensuremath{\mathit{Val}}}
\newcommand{\lex}{\ensuremath{\mathit{lex}}}
\newcommand{\cycle}{\ensuremath{C}}
\newcommand{\cycleAlt}{\ensuremath{L}}
\newcommand{\multicycle}{\ensuremath{M}}
\newcommand{\zeroVec}{\ensuremath{\mathbf{0}}}
\newcommand{\oneVec}{\ensuremath{\mathbf{1}}}
\newcommand{\parameter}{\ensuremath{N}}
\newcommand{\valueSum}{\ensuremath{\mathit{val}}}
\newcommand{\counters}{\ensuremath{\mu}}
\newcommand{\countersAlt}{\ensuremath{\chi}}
\newcommand{\transition}{\ensuremath{t}}
\newcommand{\flowMatrix}{\ensuremath{F}}
\newcommand{\rankCoeff}{\ensuremath{r}}
\newcommand{\offsets}{\ensuremath{z}}
\newcommand{\character}{\ensuremath{\mathtt{char}}}
\newcommand{\identity}{\ensuremath{\mathbf{Id}}}
\newcommand{\scc}{\ensuremath{S}}
\newcommand{\sccs}{\ensuremath{\mathit{SCCs}}}
\newcommand{\decreasing}{\ensuremath{T}}
\newcommand{\coneOp}{\ensuremath{\mathit{cone}}}
\newcommand{\cone}{\ensuremath{C}}
\newcommand{\difference}{\ensuremath{e}}
\newcommand{\result}{\ensuremath{\mathit{result}}}
\newcommand{\conSysP}{\ensuremath{P}}
\newcommand{\conSysQ}{\ensuremath{Q}}
\newcommand{\conSysA}{\ensuremath{A}}
\newcommand{\conSysB}{\ensuremath{B}}
\newcommand{\conSysC}{\ensuremath{C}}
\newcommand{\conSysD}{\ensuremath{D}}
\newcommand{\conSysI}{\ensuremath{\mathit{I}}}
\newcommand{\conSysII}{\ensuremath{\mathit{II}}}
\newcommand{\timeSAT}{\ensuremath{Z}}
\newcommand{\poly}{\ensuremath{\mathit{poly}}}
\newcommand{\bits}{\ensuremath{\mathit{L}}}
\newcommand{\timeLP}{\ensuremath{\timeSAT'}}
\newcommand{\complexity}{\ensuremath{\mathit{comp}}}
\newcommand{\uInv}{\ensuremath{\mathit{uBound}}}
\newcommand{\compBound}{\ensuremath{\mathit{cBound}}}
\newcommand{\expression}{\ensuremath{\mathit{exp}}}
\newcommand{\cycleVal}{\ensuremath{v}}
\newcommand{\nonPosTrans}{\ensuremath{K}}
\newcommand{\constant}{\ensuremath{p}}
\newcommand{\constantAlt}{\ensuremath{q}}
\newcommand{\parameterizedLP}{\ensuremath{C}}
\newcommand{\rationalLP}{\ensuremath{C}}
\newcommand{\optimal}{\ensuremath{c}}
\newcommand{\denominator}{\ensuremath{m}}
\newcommand{\emptysequence}{\ensuremath{\epsilon}}
\newcommand{\exampleCS}{\ensuremath{\mathit{run}}}
\newcommand{\exampleProg}{\ensuremath{\mathit{exp}}}

\newcommand{\rank}{\ensuremath{\mathit{rank}}}
\newcommand{\qrank}{\ensuremath{\mathit{qrank}}}
\newcommand{\affineRank}{\ensuremath{\mathit{rank}}}
\newcommand{\sumRank}{\ensuremath{\mathit{sum}}}
\newcommand{\sccRank}{\ensuremath{\mathit{combine}}}
\newcommand{\rankProc}{\ensuremath{\mathit{ranking}}}
\newcommand{\domain}{\ensuremath{W}}
\newcommand{\domainElement}{\ensuremath{a}}

\newcommand{\norm}[1]{\left\lVert #1 \right\rVert}

\begin{abstract}
Vector addition systems are an important model in theoretical computer science and have been used in a variety of areas.
In this paper, we consider vector addition systems with states over a parameterized initial configuration.
For these systems, we are interested in the standard notion of computational time complexity, i.e., we want to understand the length of the longest trace for a fixed vector addition system with states depending on the size of the initial configuration.
We show that the asymptotic complexity of a given vector addition system with states is either $\Theta(N^k)$ for some computable integer $k$, where $N$ is the size of the initial configuration, or at least exponential.
We further show that $k$ can be computed in polynomial time in the size of the considered vector addition system.
Finally, we show that $1 \le k \le 2^\dimension$, where $\dimension$ is the dimension of the considered vector addition system.
\end{abstract}

\section{Introduction}
\label{sec:intro}

Vector addition systems (VASs)~\cite{journals/jcss/KarpM69}, which are equivalent to Petri nets, are a popular model for the analysis of parallel processes~\cite{journals/eik/EsparzaN94}.
Vector addition systems with states (VASSs)~\cite{journals/tcs/HopcroftP79} are an extension of VASs with a finite control and are a popular model for the analysis of concurrent systems, because the finite control can for example be used to model shared global memory~\cite{journals/toplas/0001KW14}.
In this paper, we consider VASSs over a parameterized initial configuration.
For these systems, we are interested in the standard notion of computational time complexity, i.e., we want to understand the length of the longest execution for a fixed VASS depending on the size of the initial configuration.
VASSs over a parameterized initial configuration naturally arise in two areas:
1) \emph{The parameterized verification problem.}
For concurrent systems the number of system processes is often not known in advance, and thus the system is designed such that a template process can be instantiated an arbitrary number of times.
The problem of analyzing the concurrent system for all possible system sizes is a common theme in the literature~\cite{journals/jacm/GermanS92,journals/tcs/FinkelGRB06,conf/lata/AbdullaDB09,conf/fmcad/JohnKSVW13,journals/sigact/BloemJKKRVW16,conf/lpar/AminofRZ15,conf/icalp/AminofRZS15}.
2) \emph{Automated complexity analysis of programs.}
VASSs (and generalizations) have been used as backend in program analysis tools for automated complexity analysis~\cite{conf/cav/SinnZV14,conf/fmcad/SinnZV15,journals/jar/SinnZV17}.
The VASS considered by these tools are naturally parameterized over the initial configuration,
modelling the dependency of the program complexity on the program input.
The cited papers have proposed practical techniques but did not give complete algorithms.

Two recent papers have considered the computational time complexity of VASSs over a parameterized initial configuration.
\cite{conf/icalp/Leroux18} presents a PTIME procedure for deciding whether a VASS is polynomial or at least exponential, but does not give a precise analysis in case of polynomial complexity.
\cite{conf/lics/BrazdilCK0VZ18} establishes the precise asymptotic complexity for the special case of VASSs whose configurations are linearly bounded in the size of the initial configuration.
In this paper, we generalize both results and fully characterize the asymptotic behaviour of VASSs with polynomial complexity:
We show that the asymptotic complexity of a given VASS is either $\Theta(N^k)$ for some computable integer $k$, where $N$ is the size of the initial configuration, or at least exponential.
We further show that $k$ can be computed in PTIME in the size of the considered VASS.
Finally, we show that $1 \le k \le 2^\dimension$, where $\dimension$ is the dimension of the considered VASS.

\vspace{-0.3cm}
\subsection{Overview and Illustration of Results}

We discuss our approach on the VASS $\vass_\exampleCS$, stated in Figure~\ref{fig:intro}, which will serve as running example.
The VASS has dimension $3$ (i.e., the vectors annotating the transitions have dimension $3$) and four states $\state_1,\state_2,\state_3,\state_4$.
In this paper we will always represent vectors using a set of variables $\vars$, whose cardinality equals the dimension of the VASS.
For $\vass_\exampleCS$ we choose $\vars = \{x,y,z\}$ and use $x,y,z$ as indices for the first, second and third component of 3-dimensional vectors.
The configurations of a VASS are pairs of states and  valuations of the variables to non-negative integers.
A step of a VASS moves along a transition from the current state to a successor state, and adds the vector labelling the transition to the current valuation;
a step can only be taken if the resulting valuation is non-negative.
For the computational time complexity analysis of VASSs, we consider traces (sequences of steps) whose initial configurations consist of a valuation whose maximal value is bounded by $N$ (the parameter used for bounding the size of the initial configuration).
The computational time complexity is then the length of the longest trace whose initial configuration is bounded by $N$.
For ease of exposition, we will in this paper only consider VASSs whose control-flow graph is \emph{connected}. (For the general case, we remark that one needs to decompose a VASS into its strongly-connected components (SCCs), which can then be analyzed in isolation, following the DAG-order of the SCC decomposition; for this, one slightly needs to generalize the analysis in this paper to initial configurations with values $\Theta(N^{k_\var})$ for every variable $\var \in \vars$, where $k_\var \in \mathbb{Z}$.)
For ease of exposition, we further consider traces over arbitrary initial states (instead of some fixed initial state); this is justified because for a fixed initial state one can always restrict the control-flow graph
to the reachable states, and then the two options result in the same notion of computational complexity (up to a constant offset, which is not relevant for our asymptotic analysis).

In order to analyze the computational time complexity of a considered VASS,
our approach computes \emph{variable bounds} and \emph{transition bounds}.
A variable bound is the maximal value of a variable reachable by any trace whose initial configuration is bounded by $N$.
A transition bound is the maximal number of times a transition appears in any trace whose initial configuration is bounded by $N$.
For $\vass_\exampleCS$, our approach establishes the linear variable bound $\Theta(N)$ for $x$ and $y$, and the quadratic bound $\Theta(N^2)$ for $z$.
We note that because the variable bound of $z$ is quadratic and not linear, $\vass_\exampleCS$ cannot be analyzed by the procedure of~\cite{conf/lics/BrazdilCK0VZ18}.
Our approach establishes the bound $\Theta(N)$ for the transitions $\state_1\rightarrow \state_3$ and $\state_4 \rightarrow \state_2$, the bound $\Theta(N^2)$ for transitions $\state_1\rightarrow \state_2$, $\state_2\rightarrow \state_1$, $\state_3\rightarrow \state_4$, $\state_4\rightarrow \state_3$, and the bound $\Theta(N^3)$ for all self-loops.
The computational complexity of $\vass_\exampleCS$ is then the maximum of all transition bounds, i.e., $\Theta(N^3)$.
In general, our main algorithm (Algorithm~\ref{alg:algorithm} presented in Section~\ref{sec:algorithm}) either establishes that the VASS under analysis has at least exponential complexity or computes asymptotically precise variable and transition bounds $\Theta(N^k)$, with $k$ computable in PTIME and $1 \le k \le 2^\dimension$, where $\dimension$ is the dimension of the considered VASS.
We note that our upper bound $2^\dimension$ also improves the analysis of~\cite{conf/icalp/Leroux18}, which reports an exponential dependence on the number of transitions (and not only on the dimension).

We further state a family $\vass_\dimension$ of VASSs, which illustrate that $k$ can indeed be exponential in the dimension (the example can be skipped on first reading).
$\vass_\dimension$ uses variables $\var_{i,j}$ and consists of states $\state_{i,j}$, for $1 \le i \le \dimension$ and $j=1,2$.
We note that $\vass_\dimension$ has dimension $2\dimension$.
$\vass_\dimension$ consists of the transitions
\vspace{-0.2cm}
\begin{itemize}
  \item $\state_{i,1} \xrightarrow{\update} \state_{i,2}$, for $1 \le i \le \dimension$, with $\update(\var_{i,1}) = -1$ and $\update(\var) = 0$ for all $\var \neq \var_{i,1}$,
  \item $\state_{i,2}\xrightarrow{\update} \state_{i,1}$, for $1 \le i \le \dimension$, with $\update(\var) = 0$ for all $\var$,
  \item $\state_{i,1}\xrightarrow{\update} \state_{i,1}$, for $1 \le i \le \dimension$, with $\update(\var_{i,1}) = -1$, $\update(\var_{i,2}) = 1$,
      $\update(\var_{i+1,1}) =\update(\var_{i+1,2}) =1$ in case $i < \dimension$, and $\update(\var) = 0$ for all other $\var$,
  \item $\state_{i,2}\xrightarrow{\update} \state_{i,2}$, for $1 \le i \le \dimension$, with $\update(\var_{i,1}) = 1$, $\update(\var_{i,2}) = -1$, and $\update(\var) = 0$ for all other $\var$,
  \item $\state_{i,1}\xrightarrow{\update} \state_{i+1,1}$, for $1 \le i < \dimension$, with $\update(\var_{i,1}) = -1$ and $\update(\var) = 0$ for all $\var \neq \var_{i,1}$,
  \item $\state_{i+1,2}\xrightarrow{\update} \state_{i,2}$, for $1 \le i < \dimension$, with $\update(\var) = 0$ for all $\var$.
\end{itemize}
\vspace{-0.2cm}
$\vass_\exampleProg$ in Figure~\ref{fig:intro} depicts $\vass_\dimension$ for $\dimension=3$, where the vector components are stated in the order $\var_{1,1},\var_{1,2},\var_{2,1},\var_{2,2},\var_{3,1},\var_{3,2}$.
It is not hard to verify for all $1 \le i \le \dimension$ that $\Theta(N^{2^{i-1}})$ is the precise asymptotic variable bound for $\var_{i,1}$ and $\var_{i,2}$, that
$\state_{i,1} \rightarrow \state_{i,2}$, $\state_{i,2} \rightarrow \state_{i,1}$ and $\state_{i,1} \rightarrow \state_{i+1,1}$, $\state_{i+1,2} \rightarrow \state_{i,2}$ in case $i < \dimension$, and that $\Theta(N^{2^i})$ is the precise asymptotic transition bound for $\state_{i,1} \rightarrow \state_{i,1}$, $\state_{i,2} \rightarrow \state_{i,2}$ (Algorithm~\ref{alg:algorithm} can be used to find these bounds).

\begin{figure}[t]
\begin{tabular}{l|l}
\begin{minipage}{5cm}
\begin{tikzpicture}[scale=0.4, node distance = 1cm, auto]
    \node (t1) [xshift=2cm] {$\state_1$};
    \node (t2) [below of=t1,node distance=3cm] {$\state_2$};
    \node (t3) [right of=t1,node distance=2cm] {$\state_3$};
    \node (t4) [below of=t3,node distance=3cm] {$\state_4$};
    \path(t1) edge [line,bend left] node [left,xshift=0.1cm]
    {$\begin{pmatrix}
    0 \\
    0 \\
    -1
    \end{pmatrix}$}(t2)
    (t2) edge [loop below,  every loop/.append style={looseness= 20}] node [left]
    {$\begin{pmatrix}
    1 \\
    -1\\
    1
    \end{pmatrix}$}(t2)
    (t2) edge [line,bend left] node [left]
    {$\begin{pmatrix}
    0 \\
    0 \\
    -1
    \end{pmatrix}$} (t1)
    (t1) edge [loop above,  every loop/.append style={looseness= 20}] node [left]
    %, yshift=-0.3cm,xshift=0.1cm
    {$\begin{pmatrix}
    -1\\
    1 \\
    -1
    \end{pmatrix}$} (t1)
    (t3) edge [line,bend left] node [left,xshift=0.1cm]
    {$\begin{pmatrix}
    0 \\
    0 \\
    -1
    \end{pmatrix}$}(t4)
    (t4) edge [loop below,  every loop/.append style={looseness= 20}] node [right]
    {$\begin{pmatrix}
    1 \\
    -1 \\
    -1
    \end{pmatrix}$}
    (t4) edge [line,bend left] node [left,xshift=0.1cm]
    {$\begin{pmatrix}
    0 \\
    0 \\
    -1
    \end{pmatrix}$} (t3)
    (t3) edge [loop above,  every loop/.append style={looseness= 20}] node [right]
    {$\begin{pmatrix}
    -1 \\
    1 \\
    1
    \end{pmatrix}$}(t4)
    (t1) edge [line,bend left] node [above]
    {$\begin{pmatrix}
    -1 \\
    0 \\
    0
    \end{pmatrix}$} (t3)
    (t4) edge [line,bend left] node [below]
    {$\begin{pmatrix}
    0 \\
    0 \\
    0
    \end{pmatrix}$} (t2)
    ;
   \end{tikzpicture}
   \end{minipage}
&
\begin{minipage}{6cm}
\begin{tikzpicture}[scale=0.4, node distance = 1cm, auto]
    \node (t1) [xshift=2cm] {$\state_{1,1}$};
    \node (t2) [below of=t1,node distance=3.5cm] {$\state_{1,2}$};
    \node (t3) [right of=t1,node distance=2cm] {$\state_{2,1}$};
    \node (t4) [below of=t3,node distance=3.5cm] {$\state_{2,2}$};
    \node (t5) [right of=t3,node distance=2cm] {$\state_{3,1}$};
    \node (t6) [below of=t5,node distance=3.5cm] {$\state_{3,2}$};
    \path(t1) edge [line,bend left] node [left,xshift=0.05cm]
    {$\begin{pmatrix}
    -1\\
    0 \\
    0 \\
    0 \\
    0 \\
    0
    \end{pmatrix}$}(t2)
    (t2) edge [loop below,  every loop/.append style={looseness= 20}] node [below]
    {$\begin{pmatrix}
    1 \\
    -1\\
    0 \\
    0 \\
    0 \\
    0
    \end{pmatrix}$}(t2)
    (t2) edge [line,bend left] node [left]
    {$\begin{pmatrix}
    0 \\
    0 \\
    0 \\
    0 \\
    0 \\
    0
    \end{pmatrix}$} (t1)
    (t1) edge [loop above,  every loop/.append style={looseness= 20}] node [above]
    %, yshift=-0.3cm,xshift=0.1cm
    {$\begin{pmatrix}
    -1\\
    1 \\
    1 \\
    1 \\
    0 \\
    0
    \end{pmatrix}$} (t1)
    (t3) edge [line,bend left] node [left,xshift=0.05cm]
    {$\begin{pmatrix}
    0\\
    0 \\
    -1 \\
    0 \\
    0 \\
    0
    \end{pmatrix}$}(t4)
    (t4) edge [loop below,  every loop/.append style={looseness= 20}] node [below]
    {$\begin{pmatrix}
    0\\
    0 \\
    1 \\
    -1 \\
    0 \\
    0
    \end{pmatrix}$}
    (t4) edge [line,bend left] node [left]
    {$\begin{pmatrix}
    0 \\
    0 \\
    0 \\
    0 \\
    0 \\
    0
    \end{pmatrix}$} (t3)
    (t3) edge [loop above,  every loop/.append style={looseness= 20}] node [above]
    {$\begin{pmatrix}
    0 \\
    0 \\
    -1\\
    1 \\
    1 \\
    1
    \end{pmatrix}$}(t4)
    (t1) edge [line,bend left] node [above]
    {$\begin{pmatrix}
    -1 \\
    0 \\
    0 \\
    0 \\
    0 \\
    0
    \end{pmatrix}$} (t3)
    (t4) edge [line,bend left] node [below]
    {$\begin{pmatrix}
    0 \\
    0 \\
    0 \\
    0 \\
    0 \\
    0
    \end{pmatrix}$} (t2)
    (t5) edge [line,bend left] node [left,xshift=0.05cm]
    {$\begin{pmatrix}
    0\\
    0 \\
    0 \\
    0 \\
    -1 \\
    0
    \end{pmatrix}$}(t6)
    (t6) edge [loop below,  every loop/.append style={looseness= 20}] node [below]
    {$\begin{pmatrix}
    0\\
    0 \\
    0 \\
    0 \\
    1 \\
    -1
    \end{pmatrix}$}
    (t6) edge [line,bend left] node [left]
    {$\begin{pmatrix}
    0 \\
    0 \\
    0 \\
    0 \\
    0 \\
    0
    \end{pmatrix}$} (t5)
    (t5) edge [loop above,  every loop/.append style={looseness= 20}] node [above]
    {$\begin{pmatrix}
    0 \\
    0 \\
    0 \\
    0 \\
    -1 \\
    1
    \end{pmatrix}$}(t6)
    (t3) edge [line,bend left] node [above]
    {$\begin{pmatrix}
    0 \\
    0 \\
    -1 \\
    0 \\
    0 \\
    0
    \end{pmatrix}$} (t5)
    (t6) edge [line,bend left] node [below]
    {$\begin{pmatrix}
    0 \\
    0 \\
    0 \\
    0 \\
    0 \\
    0
    \end{pmatrix}$} (t4)
    ;
   \end{tikzpicture}
   \end{minipage}
%&\hspace{-0.5cm}

   \end{tabular}
   \caption{VASS $\vass_\exampleCS$ (left) and VASS $\vass_\exampleProg$ (right)}
   \vspace{-0.5cm}
   \label{fig:intro}
\end{figure}

\vspace{-0.3cm}
\subsection{Related Work}
A celebrated result on VASs is the EXPSPACE-completeness~\cite{Lipton76,journals/tcs/Rackoff78} of the boundedness problem.
Deciding termination for a VAS with a \emph{fixed} initial configuration can be reduced to the boundedness problem, and is therefore also EXPSPACE-complete;
this also applies to VASSs, whose termination problem can be reduced to the VAS termination problem.
In contrast, deciding the termination of VASSs for \emph{all} initial configurations is in PTIME.
It is not hard to see that non-termination over all initial configurations is equivalent to the existence of non-negative cycles (e.g., using
Dickson's Lemma~\cite{Dickson1913}).
Kosaraju and Sullivan have given a PTIME procedure for the detection of zero-cycles~\cite{conf/stoc/KosarajuS88}, which can be easily be adapted to non-negative cycles.
The existence of zero-cycles is decided by the repeated use of a constraint system in order to remove transitions that can definitely not be part of a zero-cycle.
The algorithm of Kosaraju and Sullivan forms the basis for both cited papers~\cite{conf/icalp/Leroux18,conf/lics/BrazdilCK0VZ18}, as well as the present paper.

A line of work~\cite{conf/cav/SinnZV14,conf/fmcad/SinnZV15,journals/jar/SinnZV17}
has used VASSs (and their generalizations) as backends for the automated complexity analysis of C programs.
These algorithms have been designed for practical applicability, but are not complete and no theoretical analysis of their precision has been given.
We point out, however, that these papers have inspired the Bound Proof Principle in Section~\ref{sec:upper-bound}.

\section{Preliminaries}
\label{sec:preliminaries}

\paragraph{Basic Notation.}
For a set $X$ we denote by $|X|$ the number of elements of $X$.
Let $\field$ be either $\mathbb{N}$ or $\mathbb{Z}$.
We write $\field^I$ for the set of vectors over $\field$ indexed by some set $I$.
We write $\field^{I\times J}$ for the set of matrices over $\field$ indexed by $I$ and $J$.
We write $\oneVec$ for the vector which has entry $1$ in every component.
Given $a \in \field^I$, we write $a(i) \in \field$ for the entry at line $i \in I$ of $a$, and $\norm{a} = \max_{i \in I} |a(i)|$ for the maximum absolute value of $a$.
Given $a \in \field^I$ and $J \subseteq I$, we denote by $a|_J \in \field^J$ the restriction of $a$ to $J$, i.e., we set $a|_J(i) = a(i)$ for all $i \in J$.
Given $A \in \field^{I\times J}$, we write $A(j)$ for the vector in column $j \in J$ of $A$ and $A(i,j) \in \field$ for the entry in column $i \in I$ and row $j \in J$ of $A$.
Given $A \in \field^{I \times J}$ and $K \subseteq J$, we denote by $A|_K \in \field^{I \times K}$ the restriction of $A$ to $K$, i.e., we set $A|_K(i,j) = A(i,j)$ for all $(i,j) \in I \times K$.
We write $\identity$ for the square matrix which has entries $1$ on the diagonal and $0$ otherwise.
Given $a,b \in \field^I$ we write $a + b \in \field^I$ for component-wise addition, $c \cdot a \in \field^I$ for multiplying every component of $a$ by some $c \in \field$ and $a \ge b$ for component-wise comparison.
Given $A \in \field^{I\times J}$, $B \in \field^{J \times K}$ and $x \in \field^J$, we write $AB \in \field^{I\times K}$ for the standard matrix multiplication,
$Ax \in \field^I$ for the standard matrix-vector multiplication, $A^T \in \field^{J \times I}$ for the transposed matrix of $A$ and $x^T \in \field^{1 \times J}$ for the transposed vector of $x$.

\vspace{-0.2cm}
\paragraph{Vector Addition System with States (VASS).}
Let $\vars$ be a finite set of variables.
A vector addition system with states (VASS) $\vass = (\states(\vass),\transitions(\vass))$ consists of a finite set of \emph{states} $\states(\vass)$ and a finite set of \emph{transitions} $\transitions(\vass)$, where $\transitions(\vass) \subseteq \states(\vass) \times \mathbb{Z}^\vars \times \states(\vass)$;
we call $\dimension = |\vars|$ the \emph{dimension} of $\vass$. We write $\state_1 \xrightarrow{\update} \state_2$ to denote a transition $(\state_1,\update,\state_2) \in \transitions(\vass)$;
we call the vector $\update$ the \emph{update} of transition $\state_1 \xrightarrow{\update} \state_2$.
A \emph{path} $\paath$ of $\vass$ is a finite sequence $\state_0 \xrightarrow{\update_1} \state_1 \xrightarrow{\update_2} \cdots \state_k$ with $\state_i \xrightarrow{\update_{i+1}} \state_{i+1} \in \transitions(\vass)$ for all $0 \le i <k$.
We define the \emph{length} of $\paath$ by $\length(\paath) = k$ and the \emph{value} of $\paath$ by $\valueSum(\paath) = \sum_{i \in [1,k]} \update_i$.
Let $\instances(\paath,\transition)$ be the number of times $\paath$ contains the transition $\transition$, i.e., the number of indices $i$ such that $\transition = \state_i \xrightarrow{\update_i} \state_{i+1}$.
We remark that $\length(\paath) = \sum_{\transition \in \transitions(\vass)} \instances(\paath,\transition)$ for every path $\paath$ of $\vass$.
Given a finite path $\paath_1$ and a path $\paath_2$ such that the last state of $\paath_1$ equals the first state of $\paath_2$,
we write $\paath = \paath_1 \paath_2$ for the path obtained by joining the last state of $\paath_1$ with the first state of $\paath_2$;
we call $\paath$ the \emph{concatenation} of $\paath_1$ and $\paath_2$, and $\paath_1 \paath_2$ a \emph{decomposition} of $\paath$.
We say $\paath'$ is a \emph{sub-path} of $\paath$, if there is a decomposition $\paath = \paath_1 \paath' \paath_2$ for some $\paath_1, \paath_2$.
A \emph{cycle} is a path that has the same start- and end-state.
A \emph{multi-cycle} is a finite set of cycles.
The value $\valueSum(\multicycle)$ of a multi-cycle $\multicycle$ is the sum of the values of its cycles.
$\vass$ is \emph{connected}, if for all $\state,\state' \in \states(\vass)$ there is a path from $\state$ to $\state'$.
VASS $\vass'$ is a \emph{sub-VASS} of $\vass$, if $\states(\vass') \subseteq \states(\vass)$ and $\transitions(\vass') \subseteq \transitions(\vass)$.
Sub-VASSs $\vass_1$ and $\vass_2$ are \emph{disjoint}, if $\states(\vass_1) \cap \states(\vass_2) = \emptyset$.
A \emph{strongly-connected component (SCC)} of a VASS $\vass$ is a maximal sub-VASS $\scc$ of $\vass$ such that $\scc$ is connected and $\transitions(\scc) \neq \emptyset$.

Let $\vass$ be a VASS.
The set of \emph{valuations} $\Val(\vass) = \mathbb{N}^\vars$ consists of $\vars$-vectors over the natural numbers (we assume $\mathbb{N}$ includes $0$).
The set of \emph{configurations} $\configs(\vass) = \states(\vass) \times \Val(\vass)$ consists of pairs of states and valuations.
A \emph{step} is a triple $((\state_1,\val_1), \update, (\state_2,\val_2)) \in \configs(\vass) \times \mathbb{Z}^{\dimOp(\vass)} \times \configs(\vass)$ such that $\val_2 = \val_1 + \update$ and $\state_1 \xrightarrow{\update} \state_2 \in \transitions(\vass)$.
We write $(\state_1,\val_1) \xrightarrow{\update} (\state_2,\val_2)$ to denote a step $((\state_1,\val_1), \update, (\state_2,\val_2))$ of $\vass$.
A \emph{trace} of $\vass$ is a finite sequence $\trace = (\state_0,\val_0) \xrightarrow{\update_1} (\state_1,\val_1) \xrightarrow{\update_2} \cdots (\state_k,\val_k)$ of steps.
We lift the notions of length and instances from paths to traces in the obvious way:
we consider the path $\paath = \state_0 \xrightarrow{\update_1} \state_1 \xrightarrow{\update_2} \cdots \state_k$ that consists of the transitions used by $\trace$, and set $\length(\trace) := \length(\paath)$ and
$\instances(\trace,\transition) = \instances(\paath,\transition)$, for all $\transition \in \transitions(\vass)$.
We denote by $\initBound(\trace) = \norm{\val_0}$ the maximum absolute value of the starting valuation $\val_0$ of $\trace$.
We say that $\trace$ \emph{reaches} a valuation $\val$, if $\val = \val_k$.
The \emph{complexity} of $\vass$ is the function $\complexity_\vass(N) = \sup_{\text{trace } \trace \text{ of } \vass, \initBound(\trace) \le N} \length(\trace)$, which returns for every $N \ge 0$ the supremum over the lengths of the traces $\trace$ with $\initBound(\trace) \le N$.
The \emph{variable bound} of a variable $\var \in \vars$ is the function $\vbound_\var(N) = \sup_{\text{trace } \trace \text{ of } \vass, \initBound(\trace) \le N, \trace \text{ reaches valuation } \val }  \val(\var)$, which returns for every $N \ge 0$ the supremum over the
the values of $\var$ reachable by traces $\trace$ with $\initBound(\trace) \le N$.
The \emph{transition bound} of a transition $\transition \in \transitions(\vass)$ is the function $\tbound_\transition(N) = \sup_{\text{trace } \trace \text{ of } \vass, \initBound(\trace) \le N} \instances(\trace,\transition)$, which returns for every $N \ge 0$ the supremum over the number of instances of $\transition$ in traces $\trace$ with $\initBound(\trace) \le N$.

\vspace{-0.2cm}
\paragraph{Rooted Tree.}
A \emph{rooted tree} is a connected undirected acyclic graph in which one node has been designated as the root.
We will usually denote the root by $\Root$.
We note that for every node $\node$ in a rooted tree there is a unique path of $\node$ to the root.
The \emph{parent} of a node $\node \neq \Root$ is the node connected to $\node$ on the path to the root.
Node $\node$ is a \emph{child} of a node $\node'$, if $\node'$ is the parent of $\node$.
$\node'$ is a \emph{descendent} of $\node$, if $\node$ lies on the path from $\node'$ to the root;
$\node'$ is a \emph{strict} descendent, if furthermore $\node \neq \node'$.
$\node$ is an \emph{ancestor} of $\node'$, if $\node'$ a descendent of $\node$; $\node$ is a \emph{strict} ancestor, if furthermore $\node \neq \node'$.
The \emph{distance} of a node $\node$ to the root, is the number of nodes $\neq\node$ on the path from $\node$ to the root.
We denote by $\layerF(\layer)$ the set of all nodes with the same distance $\layer$ to the root;
we remark that $\layerF(0) = \{\Root\}$.

\vspace{0.2cm}
All proofs are stated in the appendix.

\vspace{-0.2cm} 
\section{A Dichotomy Result}
\label{sec:farkas}

We will make use of the following matrices associated to a VASS throughout the paper:
Let $\vass$ be a VASS.
We define the \emph{update matrix} $\updates \in \mathbb{Z}^{\vars \times \transitions(\vass)}$ by setting $\updates(\transition) = \update$ for all transitions $\transition = (\state,\update,\state') \in \transitions(\vass)$.
We define the \emph{flow matrix} $\flowMatrix \in \mathbb{Z}^{\states(\vass) \times \transitions(\vass)}$ by setting $\flowMatrix(\state,\transition) = -1$, $\flowMatrix(\state',\transition) = 1$ for transitions $\transition = (\state,\update,\state')$ with $\state' \neq \state$, and $\flowMatrix(\state,\transition) = \flowMatrix(\state',\transition)= 0$ for  transitions $\transition = (\state,\update,\state')$ with $\state' = \state$;
in both cases we further set $\flowMatrix(\state'',\transition) = 0$ for all states $\state''$ with $\state'' \neq \state$ and $\state'' \neq \state'$.
We note that every column $\transition$ of $\flowMatrix$ either contains exactly one $-1$ and $1$ entry (in case the source and target of transition~$\transition$ are different) or only $0$ entries (in case the source and target of transition~$\transition$ are the same).

\vspace{-0.05cm}
\begin{example}
\label{ex:update-flow-matrices}
We state the update and flow matrix for $\vass_\exampleCS$ from Section~\ref{sec:intro}:\\
$\updates = \begin{pmatrix}
    -1 & 1 & -1 & 1 & 0 & 0 & 0 & 0 & -1 & 0\\
    1 & -1 & 1 & -1 & 0 & 0 & 0 & 0 & 0 & 0\\
    -1  & 1 & 1 & -1 & -1 & -1 & -1 & -1 & 0 & 0
\end{pmatrix}$,
$\flowMatrix = \begin{pmatrix}
    0 & 0 & 0 & 0 & 1 & -1 & 0 & 0 & -1 & 0\\
    0 & 0 & 0 & 0 & -1 & 1 & 0 & 0 & 0 & 1\\
    0 & 0 & 0 & 0 & 0 & 0 & 1 & -1 & 1 & 0\\
    0 & 0 & 0 & 0 & 0 & 0 & -1 & 1 & 0 & -1
\end{pmatrix}$,
with column order
$\state_1\rightarrow \state_1$, $\state_2\rightarrow \state_2$, $\state_3\rightarrow \state_3$, $\state_4\rightarrow \state_4$,
$\state_2\rightarrow \state_1$, $\state_1\rightarrow \state_2$, $\state_4\rightarrow \state_3$, $\state_3\rightarrow \state_4$, $\state_1\rightarrow \state_3$, $\state_4\rightarrow \state_2$
(from left to right) and row order $x,y,z$ for $\updates$ resp. $\state_1,\state_2,\state_3,\state_4$ for $\flowMatrix$ (from top to bottom).
\end{example}
\vspace{-0.05cm}

We now consider the constraint systems~($\conSysP$) and~($\conSysQ$), stated below, which have maximization objectives.
The constraint systems will be used by our main algorithm in Section~\ref{sec:algorithm}.
We observe that both constraint systems are always satisfiable (set all coefficients to zero) and that the solutions of both constraint systems are closed under addition.
Hence, the number of inequalities for which the maximization objective is satisfied is unique for optimal solutions of both constraint systems.
The maximization objectives can be implemented by suitable linear objective functions.
Hence, both constraint systems can be solved
in PTIME over the integers, because we can use linear programming over the rationales and then scale rational solutions to the integers by multiplying with the least common multiple of the denominators.

\vspace{0.7cm}
\begin{tabular}{|c|c|}
\hline
 {\begin{minipage}[c]{0.45\linewidth}
\vspace{0.2cm}
constraint system ($\conSysP$):

\vspace{0.2cm}
there exists $\counters \in \mathbb{Z}^{\transitions(\vass)}$ with
\begin{align}
  \updates \counters & \ge 0 \nonumber\\
  \counters & \ge 0 \nonumber\\
  \flowMatrix \counters & = 0 \nonumber
\end{align}

Maximization Objective:\\
Maximize the number of inequalities with $(\updates \counters)(\var) > 0$ and $\counters(\transition) > 0$
\vspace{0.2cm}
\end{minipage}}
  &
{\begin{minipage}[c]{0.5\linewidth}
\vspace{0.2cm}
constraint system ($\conSysQ$):

\vspace{0.2cm}
there exist
$\rankCoeff \in \mathbb{Z}^\vars,\offsets \in \mathbb{Z}^{\states(\vass)}$ with
\begin{align}
\rankCoeff & \ge  0 \nonumber\\
\offsets & \ge  0 \nonumber\\
\updates^T \rankCoeff + \flowMatrix^T \offsets & \le 0 \nonumber
\end{align}
Maximization Objective:\\
Maximize the number of inequalities with $\rankCoeff(\var) > 0$ and $(\updates^T \rankCoeff + \flowMatrix^T \offsets)(\transition) < 0$
\vspace{0.2cm}
\end{minipage}}\\
\hline
\end{tabular}
\vspace{0.7cm}

The solutions of~($\conSysP$) and~($\conSysQ$) are characterized by the following two lemmata:

\begin{lemma}[Cited from~\cite{conf/stoc/KosarajuS88}]
\label{lem:solution-is-multcycle}
$\counters \in \mathbb{Z}^{\transitions(\vass)}$ is a solution to constraint system~($\conSysP$) iff there exists a multi-cycle $\multicycle$ with $\valueSum(\multicycle) \ge 0$ and $\counters(\transition)$ instances of transition~$\transition$ for every $\transition \in \transitions(\vass)$.
\end{lemma}

\begin{lemma}[Cited from~\cite{conf/lics/BrazdilCK0VZ18}\footnote{There is no explicit lemma with this statement in~\cite{conf/lics/BrazdilCK0VZ18}, however the lemma is implicit in the exposition of Section 4 in~\cite{conf/lics/BrazdilCK0VZ18}.
We further note that~\cite{conf/lics/BrazdilCK0VZ18} does not include the constraint $\offsets \ge 0$.
However, this difference is minor and was added in order to ensure that ranking functions always return non-negative values, which is more standard than the choice of~\cite{conf/lics/BrazdilCK0VZ18}.
A proof of the lemma can be found in the appendix.}]
\label{lem:affine-ranking}
Let $\rankCoeff,\offsets$ be a solution to constraint system ($\conSysQ$).
Let $\affineRank(\rankCoeff,\offsets): \configs(\vass) \rightarrow \mathbb{N}$ be the function defined by
$\affineRank(\rankCoeff,\offsets)(\state,\val) = \rankCoeff^T \val + \offsets(\state)$.
Then, $\affineRank(\rankCoeff,\offsets)$ is a \emph{quasi-ranking function} for $\vass$, i.e., we have
\begin{enumerate}
  \item for all $(\state,\val) \in \configs(\vass)$ that $\affineRank(\rankCoeff,\offsets)(\state,\val) \ge 0;$
  \item for all transitions $\transition = \state_1 \xrightarrow{\update} \state_2 \in \transitions(\vass)$ and valuations $\val_1,\val_2 \in \Val(\vass)$ with $\val_2 = \val_1 + \update$ that
      $\affineRank(\rankCoeff,\offsets)(\state_1,\val_1) \ge \affineRank(\rankCoeff,\offsets)(\state_2,\val_2)$;
      moreover, the inequality is strict for every $\transition$ with $(\updates^T \rankCoeff + \flowMatrix^T \offsets)(\transition) < 0$.
\end{enumerate}
\end{lemma}

We now state a dichotomy between optimal solutions to constraint systems~($\conSysP$) and~($\conSysQ$),
which is obtained by an application of Farkas' Lemma.
This dichotomy is the main reason why we are able to compute the precise asymptotic complexity of VASSs with polynomial bounds.

\begin{lemma}
\label{lem:optimization-duality}
Let $\rankCoeff$ and $\offsets$ be an optimal solution to constraint system~($\conSysQ$) and let $\counters$ be an optimal solution to constraint system~($\conSysP$).
Then, for all variables $\var \in \vars$ we either have $\rankCoeff(\var) > 0$ or $(\updates \counters)(\var) \ge 1$, and for all transitions $\transition \in \transitions(\vass)$ we either have $(\updates^T \rankCoeff + \flowMatrix^T\offsets)(\transition) < 0$ or $\counters(\transition) \ge 1$.
\end{lemma}

\begin{example}
\label{ex:ranking-iteration}
Our main algorithm, Algorithm~\ref{alg:algorithm} presented in Section~\ref{sec:algorithm}, will directly use constraint systems~($\conSysP$) and~($\conSysQ$) in its first loop iteration, and adjusted versions in later loop iterations.
Here, we illustrate the first loop iteration.
We consider the running example $\vass_\exampleCS$, whose update and flow matrices we have stated in Example~\ref{ex:update-flow-matrices}.
An optimal solution to constraint systems~($\conSysP$) and~($\conSysQ$) is given by $\counters = (1 4 4 1 1 1 1 1 0 0)^T$ and $\rankCoeff = (2 2 0)^T$, $\offsets = (0 0 1 1)^T$.
The quasi-ranking function $\affineRank(\rankCoeff,\offsets)$ immediately establishes that $\tbound_\transition(N) \in O(N)$ for
$\transition=\state_1\rightarrow \state_3$ and $\transition=\state_4 \rightarrow \state_2$, because
1) $\affineRank(\rankCoeff,\offsets)$ decreases for these two transitions and does not increase for other transitions (by Lemma~\ref{lem:affine-ranking}), and because
2) the initial value of $\affineRank(\rankCoeff,\offsets)$ is bounded by $O(N)$, i.e., we have $\affineRank(\rankCoeff,\offsets)(\state,\val) \in O(N)$ for every state $\state \in \states(\vass_\exampleCS)$ and every valuation $\val$ with $\norm{\val} \le N$.
By a similar argument we get $\vbound_x(N) \in O(N)$ and $\vbound_y(N) \in O(N)$.
The exact reasoning for deriving upper bounds is given in Section~\ref{sec:upper-bound}.
From $\counters$ we can, by Lemma~\ref{lem:solution-is-multcycle}, obtain the cycles $\cycle_1 = \state_1 \rightarrow \state_2 \rightarrow \state_2 \rightarrow \state_2 \rightarrow \state_2 \rightarrow \state_2 \rightarrow \state_1 \rightarrow \state_1$ and $\cycle_2 =\state_3 \rightarrow \state_4 \rightarrow \state_4 \rightarrow \state_4 \rightarrow \state_4 \rightarrow \state_4 \rightarrow \state_4 \rightarrow \state_4$ with $\val(\cycle_1) + \val(\cycle_2) \ge (0 0 1)^T$ (*).
We will later show that the cycles $\cycle_1$ and $\cycle_2$ give rise to a family of traces that establish $\tbound_\transition(N) \in \Omega(N^2)$ for all transitions $\transition \in \transitions(\vass_\exampleCS)$ with $\transition \neq \state_1\rightarrow \state_3$ and $\transition \neq \state_4 \rightarrow \state_2$.
Here we give an intuition on the construction:
We consider a cycle $\cycle$ of $\vass_\exampleCS$ that visits all states at least once.
By (*), the updates along the cycles $\cycle_1$ and $\cycle_2$ cancel each other out.
However, the two cycles are not connected.
Hence, we execute the cycle $\cycle_1$ some $\Omega(N)$ times, then (s part of) the cycle $\cycle$, then execute $\cycle_2$ as often as $\cycle_1$, and finally the remaining part of $\cycle$; this we repeat $\Omega(N)$ times.
This construction also establishes the bound $\vbound_z(N) \in \Omega(N^2)$ because, by (*), we increase $z$ with every joint execution of $\cycle_1$ and $\cycle_2$.
The precise lower bound construction is given in Section~\ref{sec:lower-bound}.
\end{example}

\section{Main Algorithm}
\label{sec:algorithm}

Our main algorithm -- Algorithm~\ref{alg:algorithm} -- computes the complexity as well as variable and transition bounds of an input VASS $\vass$, either detecting that $\vass$ has at least exponential complexity or reporting precise asymptotic bounds for the transitions and variables of $\vass$ (up to a constant factor):
Algorithm~\ref{alg:algorithm} will compute values $\vExp(\var) \in \nats$ such that $\vbound_N(\var) \in \Theta(N^{\vExp(\var)})$ for every $\var \in \vars$ and values $\tExp(\transition) \in \nats$ such that $\tbound_N(\transition) \in \Theta(N^{\tExp(\transition)})$ for every $\transition \in \transitions(\vass)$.

\vspace{-0.2cm}
\paragraph{Data Structures.}
The algorithm maintains a rooted tree $\ctree$.
Every node $\node$ of $\ctree$ will always be labelled by a sub-VASSs $\vassF(\node)$ of $\vass$.
The nodes in the same layer of $\ctree$ will always be labelled by disjoint sub-VASS of $\vass$.
The main loop of Algorithm~\ref{alg:algorithm} will extend
$\ctree$ by one layer per loop iteration.
The variable $\layer$ always contains the next layer that is going to be added to $\ctree$.
For computing variable and transition bounds,  Algorithm~\ref{alg:algorithm} maintains the functions $\vExp: \vars \rightarrow \nats \cup \{\nobound\}$ and $\tExp: \transitions(\vass) \rightarrow \nats \cup \{\nobound\}$.

\vspace{-0.2cm}
\paragraph{Initialization.}
We assume $\updates$ to be the update matrix and $\flowMatrix$ to be the flow matrix associated to $\vass$ as discussed in Section~\ref{sec:farkas}.
At initialization, $\ctree$ consists of the root node $\Root$ and we set $\vassF(\Root) = \vass$, i.e., the root is labelled by the input $\vass$.
We initialize $\layer = 1$ as Algorithm~\ref{alg:algorithm} is going to add layer $1$ to $\ctree$ in the first loop iteration.
We initialize $\vExp(\var) = \nobound$ for all variables $\var \in \vars$ and $\tExp(\transition) = \nobound$ for all transitions $\transition \in \transitions(\vass)$.

\vspace{-0.2cm}
\paragraph{The constraint systems solved during each loop iteration.}
In loop iteration $\layer$, Algorithm~\ref{alg:algorithm} will set $\tExp(\transition) := \layer$ for some transitions $\transition$ and $\vExp(\var) := \layer$ for some variables $\var$.
In order to determine those transitions and variables, Algorithm~\ref{alg:algorithm} instantiates constraint systems~($\conSysP$) and~($\conSysQ$) from Section~\ref{sec:farkas} over the set of transitions $\unboundedTransitions
= \bigcup_{\node \in \layerF(\layer-1)} \transitions(\vassF(\node))$,
which contains all transitions associated to nodes in layer $\layer-1$ of $\ctree$.
However, instead of a direct instantiation using $\updates|_\unboundedTransitions$ and $\flowMatrix|_\unboundedTransitions$ (i.e., the restriction of $\updates$ and $\flowMatrix$ to the transitions $\unboundedTransitions$), we need to work with an extended set of variables and an extended update matrix.
We set $\vars_\extended := \{ (\var,\node) \mid \node \in \layerF(\layer - \vExp(\var)) \}$, where we set $n - \nobound = 0$ for all $n \in \nats$.
This means that we use a different copy of variable $\var$ for every node $\node$ in layer $\layer - \vExp(\var)$.
We note that for a variable $\var$ with $\vExp(\var) = \nobound$ there is only a single copy of $\var$ in $\vars_\extended$ because $\Root \in \layerF(0)$ is the only node in layer $0$.
We define the extended update matrix $\updates_\extended \in \mathbb{Z}^{\vars_\extended \times \unboundedTransitions}$ by setting $$\updates_\extended((\var,\node),\transition)
:= \left\{
    \begin{array}{cc}
        \updates(\var,\transition), & \text{if } \transition \in \transitions(\vassF(\node)), \\
        0, & \text{otherwise}. \\
    \end{array}   \right.
$$
Constraint systems~($\conSysI$) and~($\conSysII$) stated in Figure~\ref{fig:constraint-systems} can be recognized as instantiation of constraint systems~($\conSysP$) and~($\conSysQ$) with matrices $\updates_\extended$ and $\flowMatrix|_\unboundedTransitions$ and variables $\vars_\extended$, and hence the dichotomy stated in Lemma~\ref{lem:optimization-duality} holds.

We comment on the choice of $\vars_\extended$:
Setting $\vars_\extended = \{ (\var,\node) \mid \node \in \layerF(i) \}$ for any $i \le \layer - \vExp(\var)$ would result in correct upper bounds (while $i > \layer - \vExp(\var)$ would not).
However, choosing $i < \layer - \vExp(\var)$ does in general result in sub-optimal bounds because fewer variables make constraint system~($\conSysI$) easier and constraint system~($\conSysII$) harder to satisfy (in terms of their maximization objectives).
In fact, $i = \layer - \vExp(\var)$ is the optimal choice, because this choice allows us to prove corresponding lower bounds in Section~\ref{sec:lower-bound}.
We will further comment on key properties of constraint systems~($\conSysI$) and~($\conSysII$) in Sections~\ref{sec:upper-bound} and~\ref{sec:lower-bound}, when we outline the proofs of the upper resp. lower bound.

We note that Algorithm~\ref{alg:algorithm} does not use the optimal solution $\counters$ to constraint system~($\conSysI$) for the computation of the $\vExp(\var)$ and $\tExp(\transition)$, and hence the computation of the optimal solution $\counters$ could be removed from the algorithm.
The solution $\counters$ is however needed for the extraction of lower bounds in Sections~\ref{sec:lower-bound} and~\ref{sec:exponential}, and this is the reason why it is stated here.
The extraction of lower bounds is not explicitly added to the algorithm in order to not clutter the presentation.

\begin{algorithm}[t!]
\KwIn{a connected VASS $\vass$ with update matrix $\updates$ and flow matrix $\flowMatrix$}
$\ctree$ := single root node $\Root$ with $\vassF(\Root) = \vass$\;
$\layer$ := 1\;
$\vExp(\var) := \nobound$ for all variables $\var \in \vars$\;
$\tExp(\transition) := \nobound$ for all transitions $\transition \in \transitions(\vass)$\;
\Repeat{$\vExp(\var) \neq \nobound$ and $\tExp(\transition) \neq \nobound$ for all $\var \in \vars$ and $\transition \in \transitions(\vass)$}{

let $\unboundedTransitions := \bigcup_{\node \in \layerF(\layer-1)} \transitions(\vassF(\node))$\;

let $\vars_\extended := \{ (\var,\node) \mid \node \in \layerF(\layer - \vExp(\var)) \}$, where $n - \nobound = 0$ for $n \in \nats$\;

let $\updates_\extended \in \mathbb{Z}^{\vars_\extended \times \unboundedTransitions}$ be the matrix defined by\\
\hspace{2cm} $\updates_\extended((\var,\node),\transition)
= \left\{
    \begin{array}{cc}
        \updates(\var,\transition), & \text{if } \transition \in \transitions(\vassF(\node)) \\
        0, & \text{otherwise} \\
    \end{array}   \right.
$\;

find optimal solutions $\counters$ and $\rankCoeff$, $\offsets$ to constraint systems~(\conSysI) and~(\conSysII)\;

let $\decreasingTransitions := \{ \transition \in \unboundedTransitions \mid (\updates_\extended^T \rankCoeff + \flowMatrix|_\unboundedTransitions^T \offsets)(\transition) < 0 \}$\;

set $\tExp(\transition) := \layer$ for all $\transition \in \decreasingTransitions$\;
	
\ForEach{$\node \in \layerF(\layer-1)$}{
   let $\vass' := \vassF(\node)$ be the VASS associated to $\node$\;
   decompose $(\states(\vass'), \transitions(\vass') \setminus \decreasingTransitions)$ into SCCs\;
   \ForEach{ SCC $\scc$ of $(\states(\vass'), \transitions(\vass') \setminus \decreasingTransitions)$}{
       create a child $\node'$ of $\node$ with $\vassF(\node') = \scc$\;
   }

}

\ForEach{$\var \in \vars$ with $\vExp(\var) = \nobound$}{
    \lIf{$\rankCoeff(\var,\Root) > 0$}
   {
      set $\vExp(\var) := \layer$
   }
}

\If{there are no $\var \in \vars$, $\transition \in \transitions(\vass)$ with $\layer < \vExp(\var) + \tExp(\transition) < \nobound$}
{
   \Return ``$\vass$ has at least exponential complexity''
}

$\layer := \layer + 1$\;
}
\caption{Computes transition and variable bounds for a VASS $\vass$}
\label{alg:algorithm}
\end{algorithm}

\begin{figure}[t!]
%\vspace{0.3cm}
\begin{tabular}{|c|c|}
\hline
 {\begin{minipage}[c]{0.45\linewidth}
\vspace{0.2cm}
constraint system ($\conSysI$):

\vspace{0.2cm}
there exists $\counters \in \mathbb{Z}^\unboundedTransitions$ with
\begin{align}
  \updates_\extended \counters & \ge 0 \nonumber\\
  \counters & \ge 0 \nonumber\\
  \flowMatrix|_\unboundedTransitions \counters & = 0 \nonumber
\end{align}

Maximization Objective:\\
Maximize the number of inequalities with $(\updates_\extended \counters)(\var) > 0$ and $\counters(\transition) > 0$
\vspace{0.2cm}
\end{minipage}}
  &
{\begin{minipage}[c]{0.51\linewidth}
\vspace{0.2cm}
constraint system ($\conSysII$):

\vspace{0.2cm}
there exist
$\rankCoeff \in \mathbb{Z}^{\vars_\extended},\offsets \in \mathbb{Z}^{\states(\vass)}$ with
\begin{align}
\rankCoeff & \ge  0 \nonumber\\
\offsets & \ge  0 \nonumber\\
\updates_\extended^T \rankCoeff + \flowMatrix|_\unboundedTransitions^T \offsets & \le 0 \nonumber
\end{align}
Maximization Objective:\\
Maximize the number of inequalities with $\rankCoeff(\var,\node) > 0$ and $(\updates_\extended^T \rankCoeff + \flowMatrix|_\unboundedTransitions^T \offsets)(\transition) < 0$
\vspace{0.2cm}
\end{minipage}}\\
\hline
\end{tabular}
\vspace{-0.2cm}
\caption{Constraint Systems~($\conSysI$) and~($\conSysII$) used by Algorithm~\ref{alg:algorithm}}
\label{fig:constraint-systems}
\vspace{-0.3cm}
\end{figure}

\vspace{-0.3cm}
\paragraph{Discovering transition bounds.}
After an optimal solution $\rankCoeff$, $\offsets$ to constraint system~(\conSysII) has been found,
Algorithm~\ref{alg:algorithm} collects all transitions $\transition$ with $(\updates_\extended^T \rankCoeff + \flowMatrix|_\unboundedTransitions^T \offsets)(\transition) < 0$ in the set $\decreasingTransitions$ (note that the optimization criterion in constraint system~($\conSysII$) tries to find as many such $\transition$ as possible).
Algorithm~\ref{alg:algorithm} then sets $\tExp(\transition) := \layer$ for all $\transition \in \decreasingTransitions$.
The transitions in $\decreasingTransitions$ will not be part of layer $\layer$ of $\ctree$.

\vspace{-0.3cm}
\paragraph{Construction of the next layer in $\ctree$.}
For each node $\node$ in layer $\layer-1$, Algorithm~\ref{alg:algorithm} will create children by removing the transitions in $\decreasingTransitions$.
This is done as follows:
Given a node $\node$ in layer $\layer-1$, Algorithm~\ref{alg:algorithm} considers the VASS $\vass' = \vassF(\node)$ associated to $\node$.
Then, $(\states(\vass'), \transitions(\vass') \setminus \decreasingTransitions)$ is decomposed into its SCCs.
Finally, for each SCC $\scc$ of $(\states(\vass'), \transitions(\vass') \setminus \decreasingTransitions)$ a child $\node'$ of $\node$ is created with $\vassF(\node') = \scc$.
Clearly, the new nodes in layer $\layer$ are labelled by disjoint sub-VASS of $\vass$.

\vspace{-0.3cm}
\paragraph{The transitions of the next layer.}
The following lemma states that the new layer $\layer$ of $\ctree$ contains all transitions of layer $\layer-1$ except for the transitions $\decreasingTransitions$;
the lemma is due to the fact that every transition in
$\unboundedTransitions \setminus \decreasingTransitions$ belongs to a cycle and hence to some SCC that is part of the new layer $\layer$.

\begin{lemma}
\label{lem:full-decomposition}
We consider the new layer constructed during loop iteration $\layer$ of Algorithm~\ref{alg:algorithm}:
we have
$\unboundedTransitions \setminus \decreasingTransitions = \bigcup_{\node \in \layerF(\layer)} \transitions(\vassF(\node))$.
\end{lemma}

\vspace{-0.3cm}
\paragraph{Discovering variable bounds.}
For each $\var \in \vars$ with $\vExp(\var) = \nobound$, Algorithm~\ref{alg:algorithm} checks whether $\rankCoeff(\var,\Root) > 0$ (we point out that the optimization criterion in constraint systems~($\conSysII$) tries to find as many such $\var$ with $\rankCoeff(\var,\Root) > 0$ as possible).
Algorithm~\ref{alg:algorithm} then sets $\vExp(\var) := \layer$ for all those variables.

\vspace{-0.3cm}
\paragraph{The check for exponential complexity.}
In each loop iteration, Algorithm~\ref{alg:algorithm} checks whether there are $\var \in \vars$, $\transition \in \transitions(\vass)$ with $\layer < \vExp(\var) + \tExp(\transition) < \nobound$.
If this is not the case, then we can conclude that $\vass$ is at least exponential
(see Theorem~\ref{thm:exponential} below).
If the check fails, Algorithm~\ref{alg:algorithm} increments $\layer$ and continues with the construction of the next layer in the next loop iteration.

\vspace{-0.3cm}
\paragraph{Termination criterion.}
The algorithm proceeds until either exponential complexity has been detected or until $\vExp(\var) \neq \nobound$ and $\tExp(\transition) \neq \nobound$ for all $\var \in \vars$ and $\transition \in \transitions(\vass)$ (i.e., bounds have been computed for all variables and transitions).

\vspace{-0.3cm}
\paragraph{Invariants.}
We now state some simple invariants maintained by
Algorithm~\ref{alg:algorithm}, which are easy to verify:
\begin{itemize}
  \item For every node $\node$ that is a descendent of some node $\node'$ we have that $\vassF(\node)$ is a sub-VASS of $\vassF(\node')$.
  \item The value of $\vExp$ and $\tExp$ is changed at most once for each input; when the value is changed, it is changed from $\nobound$ to some value $\neq \nobound$.
  \item For every transition $\transition \in \transitions(\vass)$ and layer $\layer$ of $\ctree$, we have that either $\tExp(\transition) \le \layer$ or there is a node $\node \in \layerF(\layer)$ such that $\transition \in \transitions(\vassF(\node))$.
  \item We have $\tExp(\transition) = \layer$ for $\transition \in \transitions(\vass)$ if and only if there is a $\node \in \layerF(\layer-1)$ with $\transition \in \transitions(\vassF(\node))$ and there is no $\node \in \layerF(\layer)$ with $\transition \in \transitions(\vassF(\node))$.
\end{itemize}

\begin{figure}[t]
\begin{tabular}{l|ll}
\begin{minipage}{5.6cm}
$\updates_\extended = \begin{pmatrix}
    -1 & 1 & 0 & 0 & 0 & 0 & 0 & 0 \\
    1 & -1 & 0 & 0 & 0 & 0 & 0 & 0\\
    0 & 0 & -1 & 1 & 0 & 0 & 0 & 0 \\
    0 & 0 & 1 & -1 & 0 & 0 & 0 & 0\\
    -1  & 1 & 1 & -1 & -1 & -1 & -1 & -1
\end{pmatrix}$
\vspace{0.1cm}
with column order
$\state_1\rightarrow \state_1$, $\state_2\rightarrow \state_2$, $\state_3\rightarrow \state_3$, $\state_4\rightarrow \state_4$,
$\state_2\rightarrow \state_1$, $\state_1\rightarrow \state_2$, $\state_4\rightarrow \state_3$, $\state_3\rightarrow \state_4$
(from left to right) and row order $(x,\node_A),(y,\node_A),(x,\node_B)$, $(y,\node_B),(z,\Root)$ (from top to bottom)
\end{minipage}
&
\begin{minipage}{3.6cm}
$\updates_\extended = \begin{pmatrix}
    -1 & 0 & 0 & 0 \\
    1 & 0 & 0 & 0 \\
    0 & 1 & 0 & 0 \\
    0 & -1 & 0 & 0 \\
    0 & 0 & -1 & 0 \\
    0 & 0 & 1 & 0 \\
    0 & 0 & 0 & 1 \\
    0 & 0 & 0 & -1 \\
    -1 & 1 & 0 & 0 \\
    0 & 0 & 1 & -1
\end{pmatrix}$
\end{minipage}
&
\begin{minipage}{2.9cm}
with column order
$\state_1\rightarrow \state_1$, $\state_2\rightarrow \state_2$, $\state_3\rightarrow \state_3$, $\state_4\rightarrow \state_4$,
(from left to right) and row order $(x,\node_1),(y,\node_1),(x,\node_2)$, $(y,\node_2),(x,\node_3),(y,\node_3)$, $(x,\node_4),(y,\node_4),(z,\node_A)$, $(z,\node_B)$ (from top to bottom)
\end{minipage}
%&\hspace{-0.5cm}

   \end{tabular}
   \caption{The extended update matrices during iteration $\layer=2$ (left) and $\layer=3$ (right) of Algorithm~\ref{alg:algorithm} on the running example $\vass_\exampleCS$ from Section~\ref{sec:intro}.}
   \vspace{-0.2cm}
   \label{fig:extended-update-matrices}
\end{figure}

 \vspace{-0.2cm}
\begin{example}
We sketch the execution of Algorithm~\ref{alg:algorithm} on $\vass_\exampleCS$.
In iteration $\layer = 1$, we have $\vars_\extended = \{(x,\Root), (y,\Root), (z,\Root)\}$, and thus matrix $\updates_\extended$ is identical to the matrix $\updates$.
Hence, constraint systems~($\conSysI$) and~($\conSysII$) are identical to constraint systems~($\conSysP$) and~($\conSysQ$), whose optimal solutions $\counters = (1 4 4 1 1 1 1 1 0 0)^T$ and $\rankCoeff = (2 2 0)^T$, $\offsets = (0 0 1 1)^T$ we have discussed in Example~\ref{ex:ranking-iteration}.
Algorithm~\ref{alg:algorithm} then sets $\tExp(\state_1 \rightarrow \state_3) = 1$ and $\tExp(\state_4 \rightarrow \state_2) = 1$,
creates two children $\node_A$ and $\node_B$ of $\Root$ labeled by $\vass_A = (\{\state_1,\state_2\}, \{ \state_1\rightarrow \state_1,\state_1\rightarrow \state_2,\state_2\rightarrow \state_2,\state_2\rightarrow \state_1\})$ and $\vass_B = (\{\state_3,\state_4\}, \{ \state_3\rightarrow \state_3,\state_3\rightarrow \state_4,\state_4\rightarrow \state_4,\state_4\rightarrow \state_3\})$, and sets $\vExp(x)=1$ and $\vExp(y)=1$.
In iteration $\layer = 2$,
we have $\vars_\extended = \{(x,\node_A), (y,\node_A), (x,\node_B), (y,\node_B), (z,\Root)\}$ and the matrix $\updates_\extended$ stated in Figure~\ref{fig:extended-update-matrices}.
Algorithm~\ref{alg:algorithm} obtains $\counters = (1 1 1 1 0 0 0 0)^T$ and $\rankCoeff = (1 2 2 1 1)^T$, $\offsets = (0 0 0 0)^T$ as optimal solutions to~($\conSysI$) and~($\conSysII$).
Algorithm~\ref{alg:algorithm} then sets $\tExp(\state_1   \rightarrow \state_2) = \tExp(\state_2 \rightarrow \state_1) = \tExp(\state_3 \rightarrow \state_4) = \tExp(\state_4 \rightarrow \state_3) = 2$, creates the children $\node_1,\node_2$ resp. $\node_3,\node_4$ of $\node_A$ resp. $\node_B$ with $\node_i$ labelled by $\vass_i = (\{\state_i\},\{\state_i\rightarrow\state_i\})$, and sets $\vExp(z)=2$.
In iteration $\layer = 3$,
we have $\vars_\extended = \{(x,\node_1), (y,\node_1), (x,\node_2), (y,\node_2), (x,\node_3),$ $(y,\node_3), (x,\node_4), (y,\node_4), (z,\node_A),(z,\node_B)\}$ and the matrix $\updates_\extended$ stated in Figure~\ref{fig:extended-update-matrices}.
Algorithm~\ref{alg:algorithm} obtains $\counters = (0 0 0 0)^T$ and $\rankCoeff = (1 1 1 3 3 1 1 1 1 1)^T$, $\offsets = (0 0 0 0)^T$ as optimal solutions to~($\conSysI$) and~($\conSysII$).
Algorithm~\ref{alg:algorithm} then sets $\tExp(\state_i\rightarrow\state_i) = 3$, for all $i$, and terminates.
\end{example}

We now state the main properties of Algorithm~\ref{alg:algorithm}:

\begin{lemma}
\label{lem:termination}
  Algorithm~\ref{alg:algorithm} always terminates.
\end{lemma}

\begin{theorem}
\label{thm:exponential}
If Algorithm~\ref{alg:algorithm} returns ``$\vass$ has at least exponential complexity'', then $\complexity_\vass(N) \in 2^{\Omega(N)}$, and we have $\tbound_\transition(N) \in 2^{\Omega(N)}$ for all $\transition \in \transitions(\vass)$ with $\tExp(\transition) = \nobound$ and $\vbound_\transition(N) \in 2^{\Omega(N)}$ for all $\var \in \vars$ with $\vExp(\var) = \nobound$.
\end{theorem}
The proof of Theorem~\ref{thm:exponential} is stated in Section~\ref{sec:exponential}.
We now assume that Algorithm~\ref{alg:algorithm} does not return ``$\vass$ has at least exponential complexity''.
Then, Algorithm~\ref{alg:algorithm} must terminate with $\tExp(\transition) \neq \nobound$ and $\vExp(\var) \neq \nobound$ for all $\transition \in \transitions(\vass)$ and $\var \in \vars$.
The following result states that $\tExp$ and $\vExp$ contain the precise exponents of the asymptotic transition and variable bounds of $\vass$:

\begin{theorem}
\label{thm:bounds}
$\vbound_N(\var) \in \Theta(N^{\vExp(\var)})$ for all $\var \in \vars$ and $\tbound_N(\transition) \in \Theta(N^{\tExp(\transition)})$ for all $\transition \in \transitions(\vass)$.
\end{theorem}

The upper bounds of Theorem~\ref{thm:bounds} will be proved in Section~\ref{sec:upper-bound} (Theorem~\ref{thm:upper-bound}) and the lower bounds in Section~\ref{sec:lower-bound} (Corollary~\ref{cor:lower-bound}).

We will prove in Section~\ref{sec:characterization} that the exponents of the variable and transition bounds are bounded exponentially in the dimension of $\vass$:
\begin{theorem}
\label{thm:bounds-size}
  We have $\vExp(\var) \le 2^{|\vars|}$ for all $\var \in \vars$ and $\tExp(\transition) \le 2^{|\vars|}$ for all $\transition \in \transitions(\vass)$.
\end{theorem}

Finally, we obtain the following corollary from Theorems~\ref{thm:bounds} and~\ref{thm:bounds-size}:

\begin{corollary}
\label{cor:complexity}
  Let $\vass$ be a connected VASS.
  Then, either $\complexity_\vass(N) \in 2^{\Omega(N)}$ or $\complexity_\vass(N) \in \Theta(N^i)$ for some computable $1 \le i \le 2^{|\vars|}$.
\end{corollary}

\subsection{Complexity of Algorithm~\ref{alg:algorithm}}

In the remainder of this section we will establish the following result:

\begin{theorem}
  Algorithm~\ref{alg:algorithm} (with the below stated optimization) can be implemented in polynomial time with regard to the size of the input VASS $\vass$.
\end{theorem}

We will argue that A) every loop iteration of Algorithm~\ref{alg:algorithm} only takes polynomial time, and B) that polynomially many loop iterations are sufficient (this only holds for the optimization of the algorithm discussed below).

Let $\vass$ be a VASS, let $m = |\transitions(\vass)|$ be the number of transitions of $\vass$, and let $n = |\vars|$ be the dimension of $\vass$.
We note that $|\layerF(\layer)| \le m$ for every layer $\layer$ of $\ctree$, because the VASSs of the nodes in the same layer are disjoint.

A) Clearly, removing the decreasing transitions and computing the strongly connected components can be done in polynomial time.
It remains to argue about constraint systems~(\conSysI) and~(\conSysII).
We observe that $|\vars_\extended| = |\{ (\var,\node) \mid \node \in \layerF(\layer - \vExp(\var)) \}| \le n \cdot m$ and $|U| \le m$.
Hence the size of constraint systems~(\conSysI) and~(\conSysII) is polynomial in the size of $\vass$.
Moreover, constraint systems~(\conSysI) and~(\conSysII) can be solved in PTIME as noted in Section~\ref{sec:farkas}.

B) We do not a-priori have a bound on the number of iterations of the main loop of Algorithm~\ref{alg:algorithm}.
(Theorem~\ref{thm:bounds-size} implies that the number of iterations is at most exponential; however, we do not use this result here).
We will shortly state an improvement of Algorithm~\ref{alg:algorithm} that ensures that polynomially many iterations are sufficient.
The underlying insight is that certain layers of the tree do not need to be constructed explicitly.
This insight is stated in the lemma below:
\begin{lemma}
\label{lem:skip-lemma}
  We consider the point in time when the execution of Algorithm~\ref{alg:algorithm} reaches line $\layer := \layer + 1$ during some loop iteration $\layer \ge 1$.
  Let $\possibleBoundSet = \{ \tExp(\transition) + \vExp(\var) \mid \var \in \vars, \transition \in \transitions(\vass)\}$ and let $\layer' = \min \{ \layer' \mid \layer' > \layer, \layer' \in \possibleBoundSet\}$.
  Then, $\vExp(\var) \neq i$ and $\tExp(\transition) \neq i$ for all $\var \in \vars$, $\transition \in \transitions(\vass)$ and $\layer < i < \layer'$.
\end{lemma}

We now present the optimization that achieves polynomially many loop iterations.
We replace the line $\layer := \layer + 1$ by the two lines $\possibleBoundSet := \{ \tExp(\transition) + \vExp(\var) \mid \var \in \vars, \transition \in \transitions(\vass)
\}$ and
$\layer := \min \{ \layer' \mid \layer' > \layer , \layer' \in \possibleBoundSet\}$.
The effect of these two lines is that Algorithm~\ref{alg:algorithm} directly skips to the next relevant layer.
Lemma~\ref{lem:skip-lemma}, stated above, justifies this optimization:
First, no new variable or transition bound is discovered in the intermediate layers $\layer < i < \layer'$.
Second, each intermediate layer $\layer < i < \layer'$ has the same number of nodes as layer $\layer$, which are labelled by the same sub-VASSs as the nodes in $\layer$ (otherwise there would be a transition with transition bound $\layer < i < \layer'$);
hence, whenever needed, Algorithm~\ref{alg:algorithm} can construct a missing layer $\layer < i < \layer'$ on-the-fly from layer $\layer$.

We now analyze the number of loop iterations of the optimized algorithm.
We recall that the value of each $\vExp(\var)$ and $\tExp(\transition)$ is changed at most once from $\nobound$ to some value $\neq \nobound$.
Hence, Algorithm~\ref{alg:algorithm} encounters at most $n \cdot m$ different values in the set $\possibleBoundSet = \{ \tExp(\transition) + \vExp(\var) \mid \var \in \vars, \transition \in \transitions(\vass) \}$ during execution.
Thus, the number of loop iterations is bounded by $n \cdot m$.
\section{Proof of the Upper Bound Theorem}
\label{sec:upper-bound}

We begin by stating a proof principle for obtaining upper bounds.

\begin{proposition}[Bound Proof Principle]
\label{prop:bound-proof-principle}
Let $\vass$ be a VASS.
Let $\unboundedTransitions \subseteq \transitions(\vass)$ be a subset of the transitions of $\vass$.
Let $\coeff: \configs(\vass) \rightarrow \nats$ and $\increase_\transition: \nats \rightarrow \nats$, for every $\transition \in \transitions(\vass) \setminus \unboundedTransitions$, be functions such that for every trace $\trace = (\state_0,\val_0) \xrightarrow{\update_1} (\state_1,\val_1) \xrightarrow{\update_2} \cdots$ of $\vass$ with $\initBound(\trace) \le N$ we have for every $i \ge 0$ that
\begin{enumerate}[label=\arabic*),ref=\arabic*]
  \item $\state_i \xrightarrow{\update_i} \state_{i+1} \in \unboundedTransitions$ implies $\coeff(\state_i,\val_i) \ge \coeff(\state_{i+1},\val_{i+1})$, and
  \item $\state_i \xrightarrow{\update_i} \state_{i+1} \in \transitions(\vass) \setminus \unboundedTransitions$ implies $\coeff(\state_i,\val_i) + \increase_\transition(N) \ge \coeff(\state_{i+1},\val_{i+1})$.
\end{enumerate}
We call such a function $\coeff$ a \emph{complexity witness} and the associated $\increase_\transition$ functions the \emph{increase certificates}.

Let $\transition \in \unboundedTransitions$ be a transition on which $\coeff$ \emph{decreases}, i.e., we have $\coeff(\state_1,\val_1) \ge \coeff(\state_2,\val_2) -1$ for every step $(\state_1,\val_1) \xrightarrow{\update} (\state_2,\val_2)$ of $\vass$ with $\transition = \state_1 \xrightarrow{\update} \state_2$.
Then,
$$\tbound_\transition(N) \le \max_{(\state,\val) \in \configs(\vass), \norm{\val} \le N} \coeff(\state,\val) + \sum_{\transition' \in \transitions(\vass) \setminus \unboundedTransitions} \tbound_{\transition'}(N) \cdot \increase_{\transition'}(N).$$

Further, let $\var \in \vars$ be a variable such that $\val(\var) \le \coeff(\state,\val)$ for all $(\state,\val) \in \configs(\vass)$.
Then,
$$\vbound_\var(N) \le \max_{(\state,\val) \in \configs(\vass), \norm{\val} \le N} \coeff(\state,\val) + \sum_{\transition' \in \transitions(\vass) \setminus \unboundedTransitions} \tbound_{\transition'}(N) \cdot \increase_{\transition'}(N).$$
\end{proposition}

\paragraph{Proof Outline of the Upper Bound Theorem.}
Let $\vass$ be a VASS for which Algorithm~\ref{alg:algorithm} does not report exponential complexity.
We will prove by induction on loop iteration $\layer$ that $\vbound_N(\var) \in O(N^\layer)$ for every $\var \in \vars$ with $\vExp(\var) = \layer$ and that $\tbound_N(\transition) \in O(N^\layer)$ for every $\transition \in \transitions(\vass)$ with $\tExp(\transition) = \layer$.

We now consider some loop iteration $\layer \ge 1$.
Let $\unboundedTransitions = \bigcup_{\node \in \layerF(\layer-1)} \transitions(\vassF(\node))$ be the transitions, $\vars_\extended$ be the set of extended variables and $\updates_\extended \in \mathbb{Z}^{\vars_\extended \times \unboundedTransitions}$ be the update matrix considered by Algorithm~\ref{alg:algorithm} during loop iteration $\layer$.
Let $\rankCoeff,\offsets$ be some optimal solution to constraint system~($\conSysII$) computed by Algorithm~\ref{alg:algorithm} during loop iteration $\layer$.
The main idea for the upper bound proof is to use the quasi-ranking function from Lemma~\ref{lem:affine-ranking} as witness function for the Bound Proof Principle.
In order to apply Lemma~\ref{lem:affine-ranking} we need to consider the VASS associated to the matrices in constraint system~($\conSysII$):
Let $\vass_\extended$ be the VASS over variables $\vars_\extended$ associated to update matrix $\updates_\extended$ and flow matrix $\flowMatrix|_\unboundedTransitions$.
From Lemma~\ref{lem:affine-ranking} we get that $\affineRank(\rankCoeff,\offsets): \configs(\vass_\extended) \rightarrow \mathbb{N}$ is a quasi-ranking function for $\vass_\extended$.
We now need to relate $\vass$ to the extended VASS $\vass_\extended$ in order to be able to use this quasi-ranking function.
We do so by extending valuations over $\vars$ to valuations over $\vars_\extended$.
For every state $\state \in \states(\vass)$ and valuation $\val: \vars \rightarrow \nats$, we define the \emph{extended valuation} $\extF_\state(\val): \vars_\extended \rightarrow \nats$ by setting
\begin{displaymath}
  \extF_\state(\val)(\var,\node) = \left\{    \begin{array}{cc}
        \val(\var), & \text{if } \state \in \states(\vassF(\node)), \\
        0, & \text{otherwise.} \\
    \end{array}   \right.
\end{displaymath}

As a direct consequence from the definition of extended valuations, we have that
$(\state,\extF_\state(\val)) \in \configs(\vass_\extended)$ for all $(\state,\val) \in \configs(\vass)$, and that
$(\state_1,\extF_{\state_1}(\val_1))  \xrightarrow{\updates_\extended(\transition)} (\state_2,\extF_{\state_2}(\val_2))$ is a step of $\vass_\extended$ for every step $(\state_1,\val_1) \xrightarrow{\update} (\state_2,\val_2)$ of $\vass$ with $\state_1 \xrightarrow{\update} \state_2 \in \unboundedTransitions$.
We now define the witness function $\coeff$ by setting $$\coeff(\state,\val) = \affineRank(\rankCoeff,\offsets)(\state,\extF_\state(\val)) \quad \quad \text{ for all } (\state,\val) \in \configs(\vass).$$
We immediately get from Lemma~\ref{lem:affine-ranking} that $\coeff$ maps configurations to the non-negative integers and that condition 1) of the Bound Proof Principle is satisfied.
Indeed, we get from the first item of Lemma~\ref{lem:affine-ranking} that $\coeff(\state,\val) \ge 0$ for all $(\state,\val) \in \configs(\vass)$, and from the second item that $\coeff(\state_1,\val_1) \ge \coeff(\state_2,\val_2)$ for every step $(\state_1,\val_1) \xrightarrow{\update} (\state_2,\val_2)$ of $\vass$ with $\transition = \state_1 \xrightarrow{\update} \state_2 \in \unboundedTransitions$;
moreover, the inequality is strict if $(\updates_\extended^T \rankCoeff + \flowMatrix|_\unboundedTransitions^T \offsets)(\transition) < 0$, i.e., the witness function $\coeff$ decreases for transitions $\transition$ with $\tExp(\transition) = \layer$.
It remains to establish condition 2) of the Bound Proof Principle.
We will argue that we can find increase certificates $\increase_\transition(N) \in O(N^{\layer - \tExp(\transition)})$ for all $\transition \in \transitions(\vass) \setminus \unboundedTransitions$.
We note that $\tExp(\transition) < \layer$ for all $\transition \in \transitions(\vass) \setminus \unboundedTransitions$, and hence the induction assumption can be applied for such $\transition$.
We can then derive the desired bounds from the Bound Proof Principle because of $\sum_{\transition \in \transitions(\vass) \setminus \unboundedTransitions} \tbound_\transition(N) \cdot \increase_\transition(N) = \sum_{\transition \in \transitions(\vass) \setminus \unboundedTransitions} O(N^{\tExp(\transition)}) \cdot O(N^{\layer-\tExp(\transition)}) = O(N^\layer)$.

\begin{theorem}
\label{thm:upper-bound}
$\vbound_N(\var) \in O(N^{\vExp(\var)})$ for all $\var \in \vars$ and $\tbound_N(\transition) \in O(N^{\tExp(\transition)})$ for all $\transition \in \transitions(\vass)$.
\end{theorem}
\section{Proof of the Lower Bound Theorem}
\label{sec:lower-bound}

The following lemma will allow us to consider traces $\trace_N$ with $\initBound(\trace_N) \in O(N)$ instead of $\initBound(\trace_N) \le N$ when proving asymptotic lower bounds.

\begin{lemma}
\label{lem:lower-bound-init-valuation}
Let $\vass$ be a VASS, let $\transition \in \transitions(\vass)$ be a transition and let $\var \in \vars$ be a variable.
If there are traces $\trace_N$ with $\initBound(\trace_N) \in O(N)$ and $\instances(\trace_N,\transition) \ge N^i$, then $\tbound_N(\transition) \in \Omega(N^i)$.
If there are traces $\trace_N$ with $\initBound(\trace_N) \in O(N)$ that reach a final valuation $\val$ with $\val(\var) \ge N^i$, then $\vbound_N(\var) \in \Omega(N^i)$.
\end{lemma}

The lower bound proof uses the notion of a \emph{pre-path}, which relaxes the notion of a path:
A pre-path $\ppath = \transition_1 \cdots \transition_k$ is a finite sequence of transitions $\transition_i = \state_i \xrightarrow{\update_i} \state_i'$.
Note that we do not require for subsequent transitions that the end state of one transition is the start state of the next transition, i.e., we do not require $\state_i' = \state_{i+1}$.
We generalize notions from paths to pre-paths in the obvious way, e.g., we set $\valueSum(\ppath) = \sum_{i \in [1,k]} \update_i$ and denote by $\instances(\ppath,\transition)$, for $\transition \in \transitions(\vass)$, the number of times $\ppath$ contains the transition $\transition$.
We say the pre-path $\ppath$ \emph{can be executed from valuation} $\val$, if there are valuations $\val_i \ge 0$ with $\val_{i+1} = \val_i + \update_{i+1}$ for all $0 \le i < k$ and $\val = \val_0$;
we further say that $\ppath$ \emph{reaches} valuation $\val'$, if $\val' = \val_k$.
We will need the following relationship between execution and traces:
in case a pre-path $\ppath$ is actually a path, $\ppath$ can be executed from valuation $\val$, if and only if there is a trace with initial valuation $\val$ that uses the same sequence of transitions as $\ppath$.
Two pre-paths $\ppath = \transition_1 \cdots \transition_k$ and $\ppath' = \transition_1' \cdots \transition_l'$ can be $\emph{shuffled}$ into a pre-path $\ppath'' = \transition_1'' \cdots \transition_{k+l}''$, if $\ppath''$ is an order-preserving interleaving of $\ppath$ and $\ppath'$;
formally, there are injective monotone functions $f: [1,k] \rightarrow [1,k+l]$ and $g: [1,l] \rightarrow [1,k+l]$ with $f([1,k]) \cap g([1,l]) = \emptyset$ such that $\transition_{f(i)}'' = \transition_i$ for all $i \in [1,k]$ and $\transition_{g(i)}'' = \transition_i'$ for all $i \in [1,l]$.
Further, for $d \ge 1$ and pre-path $\ppath$, we denote by
$\ppath^d = \underbrace{\ppath\ppath \cdots \ppath}_d$ the pre-path that consists of $d$ subsequent copies of $\ppath$.

For the remainder of this section, we fix a VASS $\vass$ for which Algorithm~\ref{alg:algorithm} does not report exponential complexity and we fix the computed tree $\ctree$ and bounds $\vExp$, $\tExp$.
We further need to use the solutions to constraint system~(\conSysI) computed during the run of Algorithm~\ref{alg:algorithm}:
For every layer $\layer \ge 1$ and node $\node \in \layerF(\layer)$,
we fix a cycle $\cycle(\node)$  that contains $\counters(\transition)$ instances of every $\transition \in \transitions(\vassF(\node))$,
where $\counters$ is an optimal solution to constraint system~(\conSysI) during loop iteration $\layer$.
The existence of such cycles is stated in Lemma~\ref{lem:cycle-existence} below.
We note that this definition ensures $\valueSum(\cycle(\node)) = \sum_{\transition \in \transitions(\vassF(\node))} \updates(\transition) \cdot \counters(\transition)$.
Further, for the root node $\Root$, we fix an arbitrary cycle $\cycle(\Root)$ that uses all transitions of $\vass$ at least once.

\begin{lemma}
\label{lem:cycle-existence}
Let $\counters$ be an optimal solution to constraint system~(\conSysI) during loop iteration $\layer$ of Algorithm~\ref{alg:algorithm}.
Then there is a cycle $\cycle(\node)$ for every $\node \in \layerF(\layer)$ that contains exactly $\counters(\transition)$ instances of every transition $\transition \in \transitions(\vassF(\node))$.
\end{lemma}

\paragraph{Proof Outline of the Lower Bound Theorem.\\}
\textbf{Step I)} We define a pre-path $\ppathLayer_\layer$, for every $\layer \ge 1$, with the following properties:
\begin{enumerate}[label=\arabic*), ref=\arabic*]
  \item $\instances(\ppathLayer_\layer,\transition) \ge N^{\layer+1}$ for all transitions $\transition \in \bigcup_{\node \in \layerF(\layer)} \transitions(\vassF(\node))$. \label{enum:prop-complete-one}
  \item $\valueSum(\ppathLayer_\layer) = N^{\layer+1} \sum_{\node \in \layerF(\layer)} \valueSum(\cycle(\node))$. \label{enum:prop-complete-two}
  \item $\valueSum(\ppathLayer_\layer)(\var) \ge 0$ for every $\var \in \vars$ with $\vExp(\var) \le \layer$.
      \label{enum:prop-complete-three}
  \item $\valueSum(\ppathLayer_\layer)(\var) \ge N^{\layer+1}$ for every $\var \in \vars$ with $\vExp(\var) \ge \layer+1$. \label{enum:prop-complete-four}
  \item $\ppathLayer_\layer$ is executable from some valuation $\val$ with
      \begin{enumerate}[label=\alph*), ref=\theenumi{}\alph*]
           \item $\val(\var) \in O(N^{\vExp(\var)})$ for $\var \in \vars$ with $\vExp(\var) \le \layer$, and \label{enum:prop-complete-five-a}
           \item $\val(\var) \in O(N^\layer)$ for $\var \in \vars$ with  $\vExp(\var) \ge \layer+1$. \label{enum:prop-complete-five-b}
      \end{enumerate} \label{enum:prop-complete-five}
\end{enumerate}
The difficulty in the construction of the pre-paths $\ppathLayer_\layer$ lies in ensuring Property~\ref{enum:prop-complete-five}).
The construction of the $\ppathLayer_\layer$ proceeds along the tree $\ctree$ using that the cycles $\cycle(\node)$ have been obtained according to solutions of constraint system~($\conSysI$).

\textbf{Step II)}
It is now a direct consequence of Properties~\ref{enum:prop-complete-three})-\ref{enum:prop-complete-five})
stated above that we can choose a sufficiently large $k > 0$ such that for every $\layer \ge 0$ the pre-path $\ppathComplete_\layer = \ppathLayer_0^k \ppathLayer_1^k \cdots \ppathLayer_\layer^k$ (the concatenation of $k$ copies of each $\ppathLayer_i$, setting $\ppathLayer_0 = \cycle(\Root)^N$), can be executed from some valuation $\val$ and reaches a valuation $\val'$ with
\begin{enumerate}[label=\arabic*), ref=\arabic*]
  \item $\norm{\val} \in O(N)$,
  \item $\val'(\var) \ge k N^{\vExp(\var)}$ for all $\var \in \vars$ with $\vExp(\var) \le \layer$, and
  \item $\val'(\var) \ge k N^{\layer+1}$ for all $\var \in \vars$ with $\vExp(\var) \ge \layer+1$.
\end{enumerate}
The above stated properties for the pre-path $\ppathComplete_{\layer_{\max}}$, where $\layer_{\max}$ is the maximal layer of $\ctree$, would be sufficient to conclude the lower bound proof except that we need to extend the proof from pre-paths to proper paths.

\textbf{Step III)} In order to extend the proof from pre-paths to paths we make use of the concept of shuffling.
For all $\layer \ge 0$, we will define paths $\witComplete_\layer$ that can be obtained by shuffling the pre-paths $\ppathComplete_0, \ppathComplete_1, \ldots, \ppathComplete_\layer$.
The path $\witComplete_{\layer_{\max}}$, where $\layer_{\max}$ is the maximal layer of $\ctree$, then has the desired properties and allows to conclude the lower bound proof with the following result:

\begin{theorem}
\label{thm:lower-bound}
There are traces $\trace_N$ with $\initBound(\trace_N) \in O(N)$ such that $\trace_N$ ends in configuration $(\state_N,\val_N)$ with $\val_N(\var) \ge N^{\vExp(\var)}$ for all variables $\var \in \vars$ and we have
$\instances(\trace_N,\transition) \ge N^{\tExp(\transition)}$ for all transitions $\transition \in \transitions(\vass)$.
\end{theorem}

With Lemma~\ref{lem:lower-bound-init-valuation} we get the desired lower bounds from Theorem~\ref{thm:lower-bound}:

\begin{corollary}
\label{cor:lower-bound}
$\vbound_N(\var) \in \Omega(N^{\vExp(\var)})$ for all $\var \in \vars$ and $\tbound_N(\transition) \in \Omega(N^{\tExp(\transition)})$ for all $\transition \in \transitions(\vass)$.
\end{corollary}

\section{The Size of the Exponents}
\label{sec:characterization}

For the remainder of this section, we fix a VASS $\vass$ for which Algorithm~\ref{alg:algorithm} does not report exponential complexity and we fix the computed tree $\ctree$ and bounds $\vExp$, $\tExp$.
Additionally, we fix a vector $\offsets_\layer \in \mathbb{Z}^{\states(\vass)}$ for every layer $\layer$ of $\ctree$ and a vector $\rankCoeff_\node \in \mathbb{Z}^\vars$ for every node $\node \in \layerF(\layer)$ as follows:
Let $\rankCoeff, \offsets$ be an optimal solution to constraint system ($\conSysII$) in iteration $\layer+1$ of Algorithm~\ref{alg:algorithm}.
We then set $\offsets_\layer = \offsets$.
For every $\node \in \layerF(\layer)$ we define $\rankCoeff_\node$ by setting
$\rankCoeff_\node(\var) = \rankCoeff(\var,\node')$,
where $\node' \in \layerF(\layer - \vExp(\var))$ is the unique ancestor of $\node$ in layer $\layer - \vExp(\var)$.
The following properties are immediate from the definition:

\begin{proposition}
\label{prop:inv-set-properties}
For every layer $\layer$ of $\ctree$ and node $\node \in \layerF(\layer)$ we have:
\begin{enumerate}[label=\arabic*),ref=\arabic*]
  \item $\offsets_\layer \ge 0$ and $\rankCoeff_\node \ge 0$. \label{enum:inv-prop-one}
  \item $\rankCoeff_\node^T\update + \offsets_\layer(\state_2) - \offsets_\layer(\state_1) \le 0$ for every transition $\state_1 \xrightarrow{\update} \state_2 \in \transitions(\vassF(\node))$;
      moreover, the inequality is strict for all transitions $\transition$ with $\tExp(\transition) = \layer+1$. \label{enum:inv-prop-two}
  \item Let $\node' \in \layerF(i)$ be a strict ancestor of $\node$.
      Then, $\rankCoeff_{\node'}^T\update + \offsets_i(\state_2) - \offsets_i(\state_1) = 0$
      for every transition $\state_1 \xrightarrow{\update} \state_2 \in \transitions(\vassF(\node))$. \label{enum:inv-prop-three}
 \item For every $\var \in \vars$ with $\vExp(\var) = \layer+1$ we have $\rankCoeff_\node(\var) > 0$ and
     $\rankCoeff_\node(\var) = \rankCoeff_{\node'}(\var)$ for all $\node' \in \layerF(\layer)$. \label{enum:inv-prop-four}
 \item For every $\var \in \vars$ with $\vExp(\var) > \layer+1$ we have $\rankCoeff_\node(\var) = 0$. \label{enum:inv-prop-five}
 \item For every $\var \in \vars$ with $\vExp(\var) \le \layer$ there is an ancestor $\node' \in \layerF(i)$ of $\node$ such that $\rankCoeff_{\node'}(\var) > 0$ and $\rankCoeff_{\node'}(\var') = 0$ for all $\var'$ with $\vExp(\var') > \vExp(\var)$. \label{enum:inv-prop-six}
\end{enumerate}
\end{proposition}

For a vector $\rankCoeff \in \mathbb{Z}^\vars$, we define the \emph{potential} of $\rankCoeff$ by setting $\pot(\rankCoeff) = \max \{ \vExp(\var) \mid \var \in \vars, \rankCoeff(\var) \neq 0 \}$, where we set $\max \emptyset = 0$.
The motivation for this definition is that we have $\rankCoeff^T\val \in O(N^{\pot(\rankCoeff)})$ for every valuation $\val$ reachable by a trace $\trace$ with $\initBound(\trace) \le N$.
We will now define the \emph{potential} of
a set of vectors $Z \subseteq \mathbb{Z}^\vars$.
Let $M$ be a matrix whose columns are the vectors of $Z$ and whose rows are ordered according to the variable bounds, i.e., if the row associated to variable $\var'$ is above the row associated to variable $\var$, then we have $ \vExp(\var') \ge \vExp(\var)$.
Let $L$ be some lower triangular matrix obtained from $M$ by elementary column operations.
We now define $\pot(Z) = \sum_{\text{column } \rankCoeff \text{ of } L} \pot(\rankCoeff)$, where we set $\sum \emptyset = 0$.
We note that $\pot(Z)$ is well-defined, because the value $\pot(Z)$ does not depend on the choice of $M$ and $L$.

We next state an upper bound on potentials.
Let $\layer \ge 0$ and let $B_\layer = \{ \vExp(\var) \mid \var \in \vars, \vExp(\var) < \layer\}$ be the set of variable bounds below $\layer$.
We set $\varsum(\layer) = 1$, for $B_\layer = \emptyset$, and $\varsum(\layer) = \sum B_\layer$, otherwise.
The following statement is a direct consequence of the definitions:

\begin{proposition}
\label{prop:pot-estimation}
Let $Z \subseteq \mathbb{Z}^\vars$ be a set of vectors  such that $\rankCoeff(\var) = 0$ for all $\rankCoeff \in Z$ and $\var \in \vars$ with $\vExp(\var) > \layer$.
Then, we have $\pot(Z) \le \varsum(\layer+1)$.
\end{proposition}

We define $\pot(\node) = \pot( \{ \rankCoeff_{\node'} \mid \node' \text{ is a strict ancestor } \text{ of } \node \} )$ as the \emph{potential} of a node $\node$.
We note that $\pot(\node) \le \varsum(\layer+1)$ for every node $\node \in \layerF(\layer)$ by Proposition~\ref{prop:pot-estimation}.
Now, we are able to state the main results of this section:

\begin{lemma}
\label{lem:entering-property}
Let $\node$ be a node in $\ctree$.
Then, every trace $\trace$ with $\initBound(\trace) \le N$ enters $\vassF(\node)$ at most $O(N^{\pot(\node)})$ times,
i.e., $\trace$  contains at most $O(N^{\pot(\node)})$ transitions $\state \xrightarrow{\update} \state'$ with $\state \not\in \states(\vassF(\node))$ and $\state' \in \states(\vassF(\node))$.
\end{lemma}

\begin{lemma}
\label{lem:var-bound-relation}
For every layer $\layer$, we have that $\vExp(\var) = \layer$ resp. $\tExp(\transition) = \layer$ implies
$\vExp(\var) \le \varsum(\layer)$ resp. $\tExp(\transition) \le \varsum(\layer)$.
\end{lemma}

The next result follows from Lemma~\ref{lem:var-bound-relation} only by arithmetic manipulations and induction on $\layer$:

\begin{lemma}
\label{lem:vars-bounded}
Let $\layer$ be some layer.
Let $k$ be the number of variables $\var \in \vars$ with $\vExp(\var) < \layer$.
Then, $\varsum(\layer) \le 2^k$.
\end{lemma}
Theorem~\ref{thm:bounds-size} is then a direct consequence of Lemma~\ref{lem:var-bound-relation} and~\ref{lem:vars-bounded} (using $k \le |\vars|$).

\section{Exponential Witness}
\label{sec:exponential}

The following lemma from~\cite{conf/icalp/Leroux18} states a condition that is sufficient for a VASS to have exponential complexity\footnote{Our formalization differs from\cite{conf/icalp/Leroux18}, but it is easy to verify that our conditions a) and b) are equivalent to the conditions on the cycles in the `iteration schemes' of~\cite{conf/icalp/Leroux18}.}.
We will use this lemma to prove  Theorem~\ref{thm:exponential}:

\begin{lemma}[Lemma 10 of~\cite{conf/icalp/Leroux18}]
\label{lem:exponential}
Let $\vass$ be a connected VASS, let $U,W$ be a partitioning of $\vars$ and let $\cycle_1,\ldots,\cycle_m$ be cycles such that a) $\valueSum(\cycle_i)(\var) \ge 0$ for all $\var \in U$ and $1 \le i \le m$, and b) $\sum_i \valueSum(\cycle_i)(\var) \ge 1$ for all $\var \in W$.
Then, there is a $c > 1$ and paths $\paath_N$ such that 1) $\paath_N$ can be executed from initial valuation $N \cdot \oneVec$, 2) $\paath_N$ reaches a valuation $\val$ with $\val(\var) \ge c^N$ for all $\var \in W$ and 3) $(\cycle_i)^{c^N}$ is a sub-path of $\paath_N$ for each $1 \le i \le m$.
\end{lemma}

We now outline the proof of Theorem~\ref{thm:exponential}:
We assume that Algorithm~\ref{alg:algorithm} returned ``$\vass$ has at least exponential complexity'' in loop iteration $\layer$.
According to Lemma~\ref{lem:cycle-existence}, there are cycles $\cycle(\node)$, for every node $\node \in \layerF(\layer)$, that contain $\counters(\transition)$ instances of every transition $\transition \in \transitions(\vassF(\node))$.
One can then show that the cycles $\cycle(\node)$ and the sets $U = \{ \var \in \vars \mid \vExp(\var) \le \layer\}$, $W = \{ \var \in \vars \mid \vExp(\var) > \layer\}$ satisfy the requirements of Lemma~\ref{lem:exponential}, which establishes Theorem~\ref{thm:exponential}.

\bibliographystyle{plain}
\bibliography{main}

\appendix

\section{Proof of Lemma~\ref{lem:solution-is-multcycle}}

\begin{proof}
"$\Leftarrow$":
We consider some multi-cycle $\multicycle$ with $\counters(\transition)$ instances of each transition~$\transition$ and $\valueSum(\multicycle) \ge 0$.
Clearly, $\counters \ge 0$ because the number of instances of a transition is always non-negative.
Because of $\valueSum(\multicycle) \ge 0$ and $\valueSum(\multicycle) = \updates \counters$, we have $\updates \counters  \ge 0$.
Because $\multicycle$ is a multi-cycle we have that $\counters$ satisfies the flow constraint $\flowMatrix \counters = 0$, which encodes for every state that the number of incoming transitions equals the number of out-going transitions.

"$\Rightarrow$":
We assume that $\counters$ is a solution  to constraint system~($\conSysP$).
Hence, we have $\counters \ge 0$.
We now consider the multi-graph which contains $\counters(\transition)$ copies of every transition~$\transition$.
From the flow constraint $\flowMatrix \counters = 0$ we have that the multi-graph is balanced, i.e., the number of incoming edges equals the number of outgoing edges for every state.
It follows that every strongly connected component has an Eulerian cycle.
Each of these Eulerian cycles gives us a cycle in the original VASS.
The union of these cycles is the desired multi-cycle because we have $\updates \counters \ge 0$ by assumption. \qed
\end{proof}

\section{Proof of Lemma~\ref{lem:affine-ranking}}

\begin{proof}
The first item holds because $\rankCoeff,\offsets$ satisfy the constraints $\rankCoeff \ge 0$ and $\offsets \ge 0$ of constraint system ($\conSysQ$).

For the second item, we consider a transition $\transition = \state_1 \xrightarrow{\update} \state_2 \in \transitions(\vass)$ and valuations $\val_1,\val_2 \in \Val(\vass)$ with $\val_2 = \val_1 + \update$.
We have $\affineRank(\rankCoeff,\offsets)(\state_2,\val_2) = \rankCoeff^T \val_2 + \offsets(\state_2) = \rankCoeff^T (\val_1 + \update) + \offsets(\state_2) =
\rankCoeff^T \val_1 + \rankCoeff^T \update + \offsets(\state_2) =
\rankCoeff^T \val_1 + \offsets(\state_1) + \update^T \rankCoeff + \offsets(\state_2) - \offsets(\state_1)
\le \rankCoeff^T \val_1 + \offsets(\state_1) = \affineRank(\rankCoeff,\offsets)(\state_1,\val_1)$,
where we have the inequality because $\rankCoeff,\offsets$
satisfies the constraint $\updates^T \rankCoeff + \flowMatrix^T \offsets \le 0$ of constraint system~($\conSysQ$).
We further observe that the inequality is strict for every $\transition$ with $(\updates^T \rankCoeff + \flowMatrix^T \offsets)(\transition) < 0$. \qed
\end{proof}

\section{Proof of Lemma~\ref{lem:optimization-duality}}

The proof of Lemma~\ref{lem:optimization-duality} will be obtained by two applications of Farkas' Lemma.
We will employ the following version of Farkas' Lemma, which states that for matrices $A$,$C$ and vectors $b$,$d$, exactly one of the following statements is true:

\vspace{0.3cm}
\begin{tabular}{|c|c|}
\hline
\begin{minipage}[c]{0.4\linewidth}

\vspace{0.2cm}
there exists $x$ with

\vspace{0.2cm}
$\begin{array}{rcr}
   Ax & \ge & b \\
   Cx & = & d
 \end{array}$

\end{minipage}
  &
\begin{minipage}[c]{0.5\linewidth}
\vspace{0.2cm}

\vspace{0.2cm}
there exist $y,z$ with

\vspace{0.2cm}
$\begin{array}{rcr}
   y & \ge & 0 \\
   A^T y + C^T z & = & 0 \\
   b^T y + d^T z & > & 0
 \end{array}$
\vspace{0.2cm}
\end{minipage} \\

\hline
\end{tabular}
\vspace{0.2cm}

We now consider the constraint systems~($\conSysA_\transition$) and~($\conSysB_\transition$) stated below.
Both constraint systems are parameterized by a transition $\transition \in \transitions(\vass)$ (we note that only Equations (\ref{multcycle:eq4}) and (\ref{multcycle:eq5}) are parameterized by $\transition$).

\vspace{0.3cm}
\begin{tabular}{|c|c|}
\hline
 {\begin{minipage}[c]{0.4\linewidth}
\vspace{0.2cm}
constraint system ($\conSysA_\transition$):

\vspace{0.2cm}
there exists $\counters \in \mathbb{Z}^{\transitions(\vass)}$ with
\begin{align}
  \updates \counters & \ge 0 \nonumber\\
  \counters & \ge 0 \nonumber\\
  \flowMatrix \counters & = 0 \nonumber\\
  \counters(\transition) & \ge 1 \label{multcycle:eq4}
\end{align}
\vspace{-0.5cm}
\end{minipage}}
  &
{\begin{minipage}[c]{0.52\linewidth}
constraint system ($\conSysB_\transition$):

\vspace{0.2cm}
there exist\\
$\rankCoeff \in \mathbb{Z}^\vars,\offsets \in \mathbb{Z}^{\states(\vass)}$ with
\begin{align}
\rankCoeff & \ge  0 \nonumber\\
\offsets & \ge  0 \nonumber\\
\updates^T \rankCoeff + \flowMatrix^T \offsets & \le 0 \text{ with } < 0 \text{ in line } \transition \label{multcycle:eq5}
\end{align}
\end{minipage}}\\
\hline
\end{tabular}
\vspace{0.3cm}

We recognize constraint system~($\conSysA_\transition$) as the dual of constraint system~($\conSysB_\transition$)
in the following Lemma:

\begin{lemma}
\label{lem:ranking-or-witness}
Exactly one of the constraint systems ($\conSysA_\transition$) and ($\conSysB_\transition$) has a solution.
\end{lemma}
\begin{proof}
We fix some transition~$\transition$.
We denote by $\character_\transition \in \mathbb{Z}^{\states(\vass)}$ the vector with $\character_\transition(\transition') = 1$, if $\transition' = \transition$, and $\character_\transition(\transition') = 0$, otherwise.
Using this notation we rewrite ($\conSysA_\transition$) to the equivalent constraint system ($\conSysA_\transition'$):

\vspace{0.2cm}
\begin{tabular}{|lc|}
\hline
{\begin{minipage}[r]{0.3\linewidth}
constraint system ($\conSysA_\transition'$):
\vspace{0.8cm}
\end{minipage}}
&
{\begin{minipage}[r]{0.3\linewidth}
\vspace{-0.2cm}
\begin{eqnarray*}
  \begin{pmatrix}
  \updates \\
  \identity
  \end{pmatrix}
  \counters & \ge &
  \begin{pmatrix}
  0 \\
  \character_\transition
  \end{pmatrix}\\
  \flowMatrix \counters & = & 0
\end{eqnarray*}
\vspace{-0.6cm}
\end{minipage}}\\
\hline
\end{tabular}
\vspace{0.2cm}

Using Farkas' Lemma, we see that either ($\conSysA_\transition'$) is satisfiable or the following constraint system ($\conSysB_\transition'$) is satisfiable:

\vspace{0.3cm}
\begin{tabular}{|r|c|}
\hline
{\begin{minipage}[r]{0.47\linewidth}
\vspace{0.2cm}
\hspace{-0.15cm}
constraint system ($\conSysB_\transition'$):

\vspace{-0.7cm}
\begin{eqnarray*}
  \begin{pmatrix}
  \rankCoeff \\
  y
  \end{pmatrix} & \ge & 0 \\
  \begin{pmatrix}
  \updates \\
  \identity
  \end{pmatrix}^T
  \begin{pmatrix}
  \rankCoeff \\
  y
  \end{pmatrix} + \flowMatrix^T \offsets & = & 0 \\
  \begin{pmatrix}
  0 \\
  \character_\transition
  \end{pmatrix}^T
  \begin{pmatrix}
  \rankCoeff \\
  y
  \end{pmatrix} + 0^T \offsets & > & 0
\end{eqnarray*}
\vspace{-0.3cm}
\end{minipage}}
  &
{\begin{minipage}[c]{0.45\linewidth}
constraint system ($\conSysB_\transition'$) simplified:
\begin{eqnarray*}
  \rankCoeff & \ge & 0 \\
  y & \ge & 0 \\
  \updates^T \rankCoeff + y + \flowMatrix^T \offsets & = & 0 \\
  y(\transition) & > & 0
\end{eqnarray*}
\end{minipage}}\\
\hline
\end{tabular}
\vspace{0.3cm}

We observe that solutions of constraint system ($\conSysB_\transition'$) are invariant under shifts of $\offsets$, i.e, if $\rankCoeff$, $y$, $\offsets$ is a solution, then $\rankCoeff$, $y$, $\offsets + c \cdot \oneVec$ is also a solution for all $c \in \mathbb{Z}$  (because every row of $\flowMatrix^T$ either contains exactly one $-1$ and $1$ entry or only $0$ entries).
Hence, we can force $\offsets$ to be non-negative.
We recognize that constraint systems ($\conSysB_\transition'$) and ($\conSysB_\transition$) are equivalent. \qed
\end{proof}

We now consider the constraint systems ($\conSysC_\var$) and ($\conSysD_\var$) stated below.
Both constraint systems are parameterized by a variable $\var \in \vars$ (we note that only Equations (\ref{multcycle:eq6}) and (\ref{multcycle:eq7}) are parameterized by $\var$).

\vspace{0.3cm}
\begin{tabular}{|c|c|}
\hline
 {\begin{minipage}[c]{0.42\linewidth}
\vspace{0.2cm}
constraint system ($\conSysC_\var$):

\vspace{0.2cm}
there exists $\counters \in \mathbb{Z}^{\transitions(\vass)}$ with
\begin{align}
  \updates \counters & \ge 0 \text{ with } \ge 1 \text{ in line } \var \label{multcycle:eq6}\\
  \counters & \ge 0 \nonumber\\
  \flowMatrix \counters & = 0 \nonumber
\end{align}
\vspace{-0.5cm}
\end{minipage}}
  &
{\begin{minipage}[c]{0.5\linewidth}
\vspace{0.2cm}
constraint system ($\conSysD_\var$):

\vspace{0.2cm}
there exist
$\rankCoeff \in \mathbb{Z}^\vars,\offsets \in \mathbb{Z}^{\states(\vass)}$ with
\begin{align}
\rankCoeff & \ge  0 \nonumber\\
\offsets & \ge  0 \nonumber\\
\updates^T \rankCoeff + \flowMatrix^T \offsets & \le 0 \nonumber\\
\rankCoeff(\var) & > 0 \label{multcycle:eq7}
\end{align}
\vspace{-0.5cm}
\end{minipage}}\\
\hline
\end{tabular}
\vspace{0.3cm}

We recognize constraint system ($\conSysC_\var$) as the dual of constraint system ($\conSysD_\var$) in the following Lemma:

\begin{lemma}
\label{lem:ranking-or-witness-two}
Exactly one of the constraint systems ($\conSysC_\var$) and ($\conSysD_\var$) has a solution.
\end{lemma}
\begin{proof}
We fix some variable~$\var \in \vars$.
We denote by $\character_\var \in \mathbb{Z}^{\vars}$ the vector with $\character_\var(\var') = 1$, if $\var' = \var$, and $\character_\var(\var') = 0$, otherwise.
Using this notation we rewrite ($\conSysA_\var$) to the equivalent constraint system ($\conSysA_\var'$):

\vspace{0.2cm}
\begin{tabular}{|lc|}
\hline
{\begin{minipage}[r]{0.3\linewidth}
constraint system ($\conSysC_\var'$):
\vspace{0.8cm}
\end{minipage}}
&
{\begin{minipage}[r]{0.3\linewidth}
\vspace{-0.2cm}
\begin{eqnarray*}
  \begin{pmatrix}
  \updates \\
  \identity
  \end{pmatrix}
  \counters & \ge &
  \begin{pmatrix}
  \character_\var \\
  0
  \end{pmatrix}\\
  \flowMatrix \counters & = & 0
\end{eqnarray*}
\vspace{-0.6cm}
\end{minipage}}\\
\hline
\end{tabular}
\vspace{0.2cm}

Using Farkas' Lemma, we see that either ($\conSysC_\var'$) is satisfiable or the following constraint system ($\conSysD_\var'$) is satisfiable:

\vspace{0.3cm}
\begin{tabular}{|r|c|}
\hline
{\begin{minipage}[r]{0.47\linewidth}
\vspace{0.2cm}
\hspace{-0.15cm}
constraint system ($\conSysD_\var'$):

\vspace{-0.7cm}
\begin{eqnarray*}
  \begin{pmatrix}
  \rankCoeff \\
  y
  \end{pmatrix} & \ge & 0 \\
  \begin{pmatrix}
  \updates \\
  \identity
  \end{pmatrix}^T
  \begin{pmatrix}
  \rankCoeff \\
  y
  \end{pmatrix} + \flowMatrix^T \offsets & = & 0 \\
  \begin{pmatrix}
  \character_\var \\
  0
  \end{pmatrix}^T
  \begin{pmatrix}
  \rankCoeff \\
  y
  \end{pmatrix} + 0^T \offsets & > & 0
\end{eqnarray*}
\vspace{-0.3cm}
\end{minipage}}
  &
{\begin{minipage}[c]{0.45\linewidth}
constraint system ($\conSysB_\transition'$) simplified:
\begin{eqnarray*}
  \rankCoeff & \ge & 0 \\
  y & \ge & 0 \\
  \updates^T \rankCoeff + y + \flowMatrix^T \offsets & = & 0 \\
  \rankCoeff(\var) & > & 0
\end{eqnarray*}
\end{minipage}}\\
\hline
\end{tabular}
\vspace{0.3cm}

We observe that solutions of constraint system ($\conSysD_\var'$) are invariant under shifts of $\offsets$, i.e, if $\rankCoeff$, $y$, $\offsets$ is a solution, then $\rankCoeff$, $y$, $\offsets + c \cdot \oneVec$ is also a solution for all $c \in \mathbb{Z}$  (because every row of $\flowMatrix^T$ either contains exactly one $-1$ and $1$ entry or only $0$ entries).
Hence, we can force $\offsets$ to be non-negative.
We recognize that constraint systems ($\conSysD_\var'$) and ($\conSysD_\var$) are equivalent. \qed
\end{proof}

We are now ready to state the proof of Lemma~\ref{lem:optimization-duality}:
\begin{proof}
We consider optimal solutions $\counters$ and $\rankCoeff$,$\offsets$ to constraint systems~($\conSysP$) and~($\conSysQ$).
The claim then directly follows from Lemma~\ref{lem:ranking-or-witness} and Lemma~\ref{lem:ranking-or-witness-two}. \qed
\end{proof}

\section{Proof of Lemma~\ref{lem:full-decomposition}}

\begin{proof}
By Lemma~\ref{lem:solution-is-multcycle} there is a multi-cycle $\multicycle$ with $\counters(\transition)$ instances of every transition~$\transition \in \unboundedTransitions$.
By Lemma~\ref{lem:optimization-duality} we have $\unboundedTransitions \setminus \decreasingTransitions = \{ \transition \in \unboundedTransitions \mid \counters(\transition) \ge 1 \}$.
Hence, every transition~$\transition \in \unboundedTransitions \setminus \decreasingTransitions$ is part of a cycle that uses only edges from $\unboundedTransitions \setminus \decreasingTransitions$.
Thus, every transition~$\transition \in \unboundedTransitions \setminus \decreasingTransitions$ must belong to some SCC of $(\states(\vassF(\node)), \transitions(\vassF(\node)) \setminus \decreasingTransitions)$ for some $\node \in \layerF(\layer-1)$.
We get that $\unboundedTransitions \setminus \decreasingTransitions = \bigcup_{\node \in \layerF(\layer)} \transitions(\vassF(\node))$.
\end{proof}

\section{Proof of Lemma~\ref{lem:termination}}

\begin{proof}
  We note that $\layer$ is incremented in every iteration of Algorithm~\ref{alg:algorithm}.
  Hence, if the values of $\vExp$ and $\tExp$ are not changed, then the condition `there are no $\transition \in \transitions(\vass)$, $\var \in \vars$ with $\layer < \tExp(\transition) + \vExp(\var) < \nobound$' will eventually become true.
  Hence, Algorithm~\ref{alg:algorithm} either terminates after finitely many iterations or there is a change in the values of $\vExp$ and $\tExp$.
  Now we recall that the value of $\vExp(\var)$ and $\tExp(\transition)$ is changed at most once for every $\transition \in \transitions(\vass)$ and $\var \in \vars$.
  Hence, Algorithm~\ref{alg:algorithm} must terminate after finitely many iterations. \qed
\end{proof}

\section{Proof of Corollary~\ref{cor:complexity}}

\begin{proof}
  We consider the run of Algorithm~\ref{alg:algorithm} on $\vass$.
  In case Algorithm~\ref{alg:algorithm} returns ``$\vass$ has at least exponential complexity'',
  then $\complexity_\vass(N) \in 2^{\Omega(N)}$ by Theorem~\ref{thm:exponential}.
  Otherwise, we have $\tExp(\transition) \neq \nobound$ and $\vExp(\var) \neq \nobound$ for all $\transition \in \transitions(\vass)$ and $\var \in \vars$.
  Let $i = \max_{\transition \in \transitions(\vass)} \tExp(\transition)$.
  Using Theorem~\ref{thm:bounds} we get that $\length(\trace) = \sum_{\transition \in \transitions(\vass)} \instances(\trace,\transition) \in O(N^i)$ for every trace $\trace$ of $\vass$ with $\initBound(\trace) \le N$.
  Hence, $\complexity_\vass(N) \in O(N^i)$.
  For the lower bound, we consider a transition $\transition \in \transitions(\vass)$ with $i = \tExp(\transition)$.
  From Theorem~\ref{thm:bounds} we get that there are traces $\trace_N$ of $\vass$ with $\initBound(\trace_N) \le N$ and $\instances(\trace_N,\transition) \in \Omega(N^i)$.
  Because of $\instances(\trace_N,\transition) \le \length(\trace_N)$ for all $N \ge 0$, we get
  $\complexity_\vass(N) \in \Omega(N^i)$. Thus, we have shown $\complexity_\vass(N) \in \Theta(N^i)$.
  Finally, $i \le 2^{|\vars|}$ by Theorem~\ref{thm:bounds-size}.
  \qed
\end{proof}

\section{Proof of Lemma~\ref{lem:skip-lemma}}

\begin{proof}
  We consider the point in time when the execution of Algorithm~\ref{alg:algorithm} reaches line $\layer := \layer + 1$ during some loop iteration $\layer \ge 1$.
  Let $\possibleBoundSet = \{ \tExp(\transition) + \vExp(\var) \mid \var \in \vars, \transition \in \transitions(\vass)\}$ and let $\layer' = \min \{ \layer' \mid \layer' > \layer, \layer' \in \possibleBoundSet\}$.

  We begin by stating the main consequence of the definition of $\layer'$:
  For every variable $\var \in \vars$, every node $\node \in \layerF(\layer - \vExp(\var))$ and every layer $\layer < i < \layer'$ there is a node $\node' \in \layerF(i - \vExp(\var))$ such that $\vassF(\node) = \vassF(\node')$ (+).
  Assume that this is not the case.
  Then there is a variable $\var \in \vars$, a node $\node \in \layerF(\layer - \vExp(\var))$ and a transition $\transition \in \transitions(\vassF(\node))$ such that
  $\layer - \vExp(\var) < \tExp(\transition) < \layer'  - \vExp(\var)$.
  However, this implies $\layer < \tExp(\transition) + \vExp(\var) < \layer'$, which contradicts the definition of $\layer'$.

  Let $\unboundedTransitions = \bigcup_{\node \in \layerF(\layer-1)} \transitions(\vassF(\node))$ be the transitions, $\vars_\extended$ be the variables
  and let $\updates_\extended$ be the update matrix considered by Algorithm~\ref{alg:algorithm} during loop iteration $\layer$.
  Let $\counters$ and $\rankCoeff$, $\offsets$ be optimal solutions to constraint systems~(\conSysI) and~(\conSysII), and let $\decreasingTransitions = \{ \transition \in \unboundedTransitions \mid (\updates_\extended^T \rankCoeff + \flowMatrix|_\unboundedTransitions^T \offsets)(\transition) < 0 \}$ be the transitions removed during loop iteration $\layer$.
  We set $\unboundedTransitions_\circ = \unboundedTransitions \setminus \decreasingTransitions$.
  By Lemma~\ref{lem:full-decomposition}, we have
  $\unboundedTransitions_\circ = \bigcup_{\node \in \layerF(\layer)} \transitions(\vassF(\node))$.
  We define $\counters_\circ$ and $\updates_\circ$ as the restriction of $\counters$ and $\updates_\extended$ to $\unboundedTransitions_\circ$, i.e.,
  we set $\counters_\circ = \counters|_{\unboundedTransitions_\circ}$ and $\updates_\circ = \updates_\extended|_{\unboundedTransitions_\circ}$.
  From Lemma~\ref{lem:optimization-duality} we get that $\unboundedTransitions_\circ = \{ \transition \mid \counters(\transition) \ge 1 \}$,
  and hence $\updates_\extended \counters = \updates_\circ \counters_\circ$.
  From this and the fact that $\counters \in \mathbb{Z}^\unboundedTransitions$ is a solution to constraint system ($\conSysI$) we get that $\updates_\circ \counters_\circ \ge 0$, $\counters_\circ \ge \oneVec$ and
  $\flowMatrix|_{\unboundedTransitions_\circ} \counters_\circ = 0$ (*).
  From Lemma~\ref{lem:optimization-duality} we further get
  that $(\updates_\circ \counters_\circ)(\var) = (\updates_\extended \counters)(\var) \ge 1$ for all variables $\var$ with $\vExp(\var) > \layer$ (\#).

  We now consider the layers $\layer \le i < \layer'$.
  We will show by induction that $\unboundedTransitions_\circ = \bigcup_{\node \in \layerF(i-1)} \transitions(\vassF(\node))$ for all $\layer \le i < \layer'$, and that $\vExp(\var) \neq i$ and $\tExp(\transition) \neq i$ for all $\layer < i < \layer'$, $\var \in \vars$ and $\transition \in \transitions(\vass)$.
  For $i = \layer$ the statement trivially holds.
  We now consider some layer $\layer < i < \layer'$.
  Layer $i$ is constructed from layer $i-1$ during loop iteration $i$.
  By induction assumption we have $\unboundedTransitions_\circ = \bigcup_{\node \in \layerF(i-1)} \transitions(\vassF(\node))$.
  During iteration $i$, thus, Algorithm~\ref{alg:algorithm} considers the update matrix
  $\updates_i \in \mathbb{Z}^{\vars_i \times \unboundedTransitions_\circ}$ for the set of (extended) variables $\vars_i = \{ (\var,\node) \mid \node \in \layerF(i - \vExp(\var)) \}$.
  From (+) we have that there is a bijective function $f: \vars_\extended \rightarrow \vars_i$ with $f(\var,\node) = (\var',\node')$ if and only if $\var = \var'$ and $\vassF(\node) = \vassF(\node')$.
  Hence, we have $f(\updates_\circ) = \updates_i$
  \footnote{For a function $f: I \rightarrow K$ and a matrix $A \in\field^{I\times J}$ we denote by $f(A) \in \field^{K \times J}$ the matrix defined by $f(A)(i,j) = A(f(i),j)$ for all $i,j \in I \times J$.}, i.e., the matrices $\updates_i$ and $\updates_\circ$ are identical up to renaming of the variables.
  Thus, from (*) we get that $\updates_i \counters_\circ \ge 0$, $\counters_\circ \ge \oneVec$ and
  $\flowMatrix|_{\unboundedTransitions_\circ} \counters_\circ = 0$, i.e., we have that
  $\counters_\circ$ is a solution to constraint system~(\conSysI) during loop iteration $i$.
  Hence, by Lemma~\ref{lem:optimization-duality} we have that
  $\updates_i^T \rankCoeff + \flowMatrix|_{\unboundedTransitions_\circ}^T \offsets = 0$ for every optimal solution $\rankCoeff$, $\offsets$ to constraint system~(\conSysII), i.e.,
  no transition is removed during loop iteration $i$.
  Thus, we get $\unboundedTransitions_\circ = \bigcup_{\node \in \layerF(i)} \transitions(\vassF(\node))$ and $\tExp(\transition) \neq i$ for all $\transition \in \transitions(\vass)$.
  Further, by (\#) we have that
  $(\updates_i \counters_\circ)(\var) \ge 1$ for all variables $\var$ with $\vExp(\var) > \layer$.
  Again by Lemma~\ref{lem:optimization-duality} we
  have for every variable $\var$ with $\vExp(\var) > \layer$ and every optimal solution $\rankCoeff$, $\offsets$ to constraint system~(\conSysII) that
  $\rankCoeff(\var,\Root) = 0$, i.e., $\vExp(\var) \neq i$ for all $\var \in \vars$.
  \qed
\end{proof}

\section{Proof of Theorem~\ref{thm:upper-bound}}

\begin{proof}
We prove the claim by induction on loop iteration $\layer$ of Algorithm~\ref{alg:algorithm}.
We consider some $\layer \ge 1$.
Let $\unboundedTransitions = \bigcup_{\node \in \layerF(\layer-1)} \transitions(\vassF(\node))$ be the transitions, $\vars_\extended$ be the set of extended variables and $\updates_\extended \in \mathbb{Z}^{\vars_\extended \times \unboundedTransitions}$ be the update matrix considered by Algorithm~\ref{alg:algorithm} during loop iteration $\layer$.
For every transition $\transition \in \transitions(\vass) \setminus \unboundedTransitions$ we have $\tExp(\var) < \layer$, and hence we can assume $\tbound_N(\transition) \in O(N^{\tExp(\transition)})$ by the induction assumption.
Further, we can apply the induction assumption for variables $\var \in \vars$ with $\vExp(\var) < \layer$ and assume $\vbound_N(\var) \in O(N^{\vExp(\var)})$.
Let $\rankCoeff,\offsets$ be some optimal solution to constraint system~($\conSysII$) computed by Algorithm~\ref{alg:algorithm} during loop iteration $\layer$.
As discussed earlier, we define the witness function $\coeff$ using the quasi-ranking function from Lemma~\ref{lem:affine-ranking}.
We note that we have $$\coeff(\state,\val) = \affineRank(\rankCoeff,\offsets)(\state,\extF_\state(\val)) = \rankCoeff^T \extF_\state(\val) + \offsets(\state) \quad \quad \text{ for all } (\state,\val) \in \configs(\vass).$$
We have already argued that $\coeff$ maps configurations to the non-negative integers, that condition 1) of the Bound Proof Principle is satisfied,
and that the witness function $\coeff$ decreases for transitions $\transition$ with $\tExp(\transition) = \layer$ (*).
It remains to establish condition 2) of the Bound Proof Principle.
We will show that there are increase certificates $\increase_\transition(N) \in O(N^{\layer - \tExp(\transition)})$ for all transitions $\transition \in \transitions(\vass) \setminus \unboundedTransitions$.

We fix some transition $\transition \in \transitions(\vass) \setminus \unboundedTransitions$.
Let $(\state_1,\val_1) \xrightarrow{\update} (\state_2,\val_2)$ be a step with $\state_1 \xrightarrow{\update} \state_2 = \transition$ in a trace $\trace$ of $\vass$ with $\initBound(\trace) \le N$.
We note that $\val_2 = \val_1 + \update$.
We will now show that $\coeff(\state_2,\val_2) - \coeff(\state_1,\val_1) \in O(N^{\layer-\tExp(\transition)})$, which is sufficient to conclude that there is an increase certificate $\increase_\transition$ with $\increase_\transition(N) \in O(N^{\layer-\tExp(\transition)})$.

Let $(\var,\node) \in \vars_\extended$ be a variable with $\vExp(\var) \le \layer-\tExp(\transition)$.
For both $i=1$ or $i=2$,
we consider the extended valuation $\extF_{\state_i}(\val_i)(\var,\node)$.
In case of $\state_i \in \states(\vassF(\node))$ we have $\extF_{\state_i}(\val_i)(\var,\node) = \val_i(\var) \in O(N^{\vExp(\var)})$ because of $\vExp(\var) \le \layer-\tExp(\transition) < \layer$, using the induction assumption for $\var$.
In case of $\state_i \not\in \states(\vassF(\node))$ we have $\extF_{\state_i}(\val_i)(\var,\node) = 0$.
The case analysis allow us to conclude that
$\extF_{\state_2}(\val_2)(\var,\node)-\extF_{\state_1}(\val_1)(\var,\node) \in O(N^{\vExp(\var)})$.

Let $(\var,\node) \in \vars_\extended$  be a variable with $\vExp(\var) > \layer - \tExp(\transition)$.
We note there is a unique node $\node' \in \layerF(\layer - \vExp(\var))$ in layer such that $\transition \in \transitions(\vassF(\node))$ (because of $\tExp(\transition) > \layer-\vExp(\var)$).
In case of $\node = \node'$, we have $\state_1,\state_2 \in \states(\vassF(\node))$, and hence $\extF_{\state_2}(\val_2)(\var,\node) =
\val_2(\var) = \val_1(\var) + \update(\var) =
\extF_{\state_1}(\val_1)(\var,\node) + \update(\var)$.
In case of $\node \neq \node'$, we have
$\state_1,\state_2 \not\in \states(\vassF(\node))$ because $\vassF(\node)$ is disjoint from $\vassF(\node')$, and hence
$\extF_{\state_2}(\val_2)(\var,\node) = \extF_{\state_1}(\val_1)(\var,\node) = 0$.
In both cases we get
$\extF_{\state_2}(\val_2)(\var,\node)-\extF_{\state_1}(\val_1)(\var,\node) \in O(1)$.

Using the above stated facts, we obtain
\begin{align*}
  & \coeff(\state_2,\val_2) - \coeff(\state_1,\val_1) = \rankCoeff^T \extF_{\state_2}(\val_2) + \offsets(\state_2) - \rankCoeff^T \extF_{\state_1}(\val_1)  - \offsets(\state_1) \\
  & = \sum_{(\var,\node) \in \vars_\extended, \vExp(\var) \le \layer-\tExp(\transition)} \rankCoeff(\var,\node) \cdot  (\extF_{\state_2}(\val_2)(\var,\node)-\extF_{\state_1}(\val_1)(\var,\node))\\
  & + \sum_{(\var,\node) \in \vars_\extended, \vExp(\var) > \layer-\tExp(\transition)} \rankCoeff(\var,\node)\cdot  (\extF_{\state_2}(\val_2)(\var,\node)-\extF_{\state_1}(\val_1)(\var,\node)) \\
  & + \offsets(\state_2) - \offsets(\state_1) \\
  & = \sum_{(\var,\node) \in \vars_\extended, \vExp(\var) \le \layer-\tExp(\transition)} O(N^{\vExp(\var)}) \\
  & + \sum_{(\var,\node) \in \vars_\extended, \vExp(\var) > \layer-\tExp(\transition)} O(1)\\
  & + \offsets(\state_2) - \offsets(\state_1) \\
  & = O(N^{\layer-\tExp(\transition)}) + O(1) + O(1) = O(N^{\layer-\tExp(\transition)}).
\end{align*}

We are now ready to apply the Bound Proof Principle from Proposition~\ref{prop:bound-proof-principle}.
We observe that $$\max_{(\state,\val) \in \configs(\vass), \norm{\val} \le N} \coeff(\state,\val) \in O(N)$$ because $\coeff(\state,\val) = \rankCoeff^T \extF_\state(\val) + \offsets(\state)$ is a linear expression for all $\state \in \states(\vass)$, and we consider valuations $\val$ with $\norm{\val} \le N$.
Further, by the above, we have
$$\sum_{\transition \in  \transitions(\vass) \setminus \unboundedTransitions} \tbound_\transition(N) \cdot \increase_\transition(N) = \sum_{\transition \in \transitions(\vass) \setminus \unboundedTransitions} O(N^{\tExp(\transition)}) \cdot O(N^{\layer-\tExp(\transition)}) = O(N^\layer),$$
using the induction assumption for $\transition \in \transitions(\vass) \setminus \unboundedTransitions$.
With (*) we can now conclude from the Bound Proof Principle that $\tbound_N(\transition) \in O(N^{\layer})$ for all transitions $\transition$ with $\tExp(\transition) = \layer$.
Next, we argue that we can also deduce the desired variable bounds.
We recall that for each variable $\var$ with $\vExp(\var) = \layer$ we have $\rankCoeff(\var,\Root) > 0$.
Hence, $\coeff(\state,\val) = \rankCoeff^T \extF_\state(\val) + \offsets(\state) \ge \rankCoeff(\var,\Root) \cdot \val(\var) \ge \val(\var)$ for all $(\state,\val) \in \configs(\vass)$.
Thus, we can conclude from the Bound Proof Principle that $\vbound_N(\var) \in O(N^{\layer})$ for all variables $\var$ with $\vExp(\var) = \layer$.  \qed
\end{proof}

\section{Proof of Lemma~\ref{lem:lower-bound-init-valuation}}

\begin{proof}
Assume that there is a $c > 0$ and that there are traces $\trace_N$ with $\initBound(\trace_N) \le cN$ and $\instances(\paath_N,\transition) \ge N^i$.
We set $N' = cN$.
We get that there are traces $\trace_{N'}$ with $\initBound(\trace_{N'}) \le N'$ and $\instances(\trace_{N'},\transition) \ge (\frac{1}{c})^i N'^i$.
Hence, $\tbound_N(\transition) \in \Omega(N^i)$.
The second claim can be shown analogously. \qed
\end{proof}

\section{Proof of Lemma~\ref{lem:cycle-existence}}
\begin{proof}
Let $\unboundedTransitions \setminus \decreasingTransitions = \bigcup_{\node \in \layerF(\layer)} \transitions(\vassF(\node))$ be the transition of the layer constructed during loop iteration $\layer$ of Algorithm~\ref{alg:algorithm}.
As we have argued in the proof of Lemma~\ref{lem:full-decomposition},
we have $\unboundedTransitions \setminus \decreasingTransitions = \{ \transition \mid \counters(\transition) \ge 1 \}$ and there is a multi-cycle $\multicycle$ with $\counters(\transition)$ instances of every transition~$\transition \in \unboundedTransitions$.
Because the $\vassF(\node)$ of the nodes $\node$ in layer $\layer$ are disjoint, we have that the transitions of every cycle of $\multicycle$ belong to only a single set $\transitions(\vassF(\node))$ for some $\node \in \layerF(\layer)$.
We can now shuffle all cycles that use transitions from the same $\vassF(\node)$ into a single cycle.
We obtain a single cycle for each $\node \in \layerF(\layer)$ that uses exactly $\counters(\transition)$ instances of every transition $\transition \in \transitions(\vassF(\node))$.
\qed
\end{proof}

\section{Proof of Theorem~\ref{thm:lower-bound}}

We will need the following two properties about pre-paths:

\begin{proposition}
\label{prop:merge-proposition}
Let $\ppath$ be a pre-path that can be obtained by shuffling the two pre-paths $\ppath_a$ and $\ppath_b$.
If $\ppath_a$ resp. $\ppath_b$ are executable from some valuation $\val_a$ resp. $\val_b$,
then $\ppath$ is executable from valuation $\val_a + \val_b$;
moreover, if $\ppath_a$ reaches a valuation $\val_a'$ from $\val_a$ and $\ppath_b$ reaches a valuation $\val_b'$ from $\val_b$, then $\ppath$ reaches valuation $\val_a' + \val_b'$ from $\val_a + \val_b$.
\end{proposition}
\begin{proof}
Let
$\ppath_a = \transition_1 \cdots \transition_k$ resp. $\ppath_b = \transition_1' \cdots \transition_l'$
be the sequences of transitions $\transition_i = \state_i \xrightarrow{\update_i} \state_{i+1}$ resp.
$\transition_i' = \state'_i \xrightarrow{\update_i'} \state_{i+1}'$.
By assumption there are injective monotone functions $f: [1,k] \rightarrow [1,k+l]$ and $g: [1,l] \rightarrow [1,k+l]$ with $f([1,k]) \cap g([1,l]) = \emptyset$ that define $\ppath = \transition_1'' \cdots \transition_{k+l}''$, i.e., we have $\transition_{f(i)}'' = \transition_i$ for all $i \in [1,k]$ and $\transition_{g(i)}'' = \transition_i'$ for all $i \in [1,l]$.
We define functions $u(i) = \max \{ h \in [1,k] \mid f(h) \le i\}$ and $v(i) = \max \{ h \in [1,l] \mid g(h) \le i\}$ for all $0 \le i \le l+k$, where we set $\max \emptyset = 0$.
Because $\ppath_a$ resp. $\ppath_b$ are executable from some valuation $\val_a$ resp. $\val_b$,
there are valuations $\val_i \ge 0$ and $\val_i' \ge 0$ such that $\val_a = \val_0$ and $\val_{i+1} = \val_i + \update_{i+1}$ for all $0 \le i < k$ as well as $\val_b = \val_0'$ and $\val_{i+1}' = \val_i' + \update_{i+1}' \ge 0$ for all $0 \le i < l$.
We now set $\val_i'' = \val_{u(i)} + \val_{v(i)}'$ for all $0 \le i \le k+l$.
We observe $\val_0'' = \val_{u(0)} + \val_{v(0)}' = \val_a + \val_b$ and $\val_{k+l}'' = \val_{u(k+l)} + \val_{v(k+l)}' = \val_a' + \val_b'$.
Clearly, $\val_i'' \ge 0$ for all $1 \le i \le k+l$.
We now observe that either $u(i) = u(i+1)$ and $v(i) \neq v(i+1)$ or $u(i) \neq u(i+1)$ and $v(i) = v(i+1)$ for all $0 \le i < k+l$.
In particular,
$\val_{i+1}'' =
\val_{u(i+1)} + \val_{v(i+1)}' =
\val_{u(i)} + \update_{u(i+1)} + \val_{v(i)}' =
\val_i'' + \update_{u(i+1)}$ or
$\val_{i+1}'' =
\val_{u(i+1)} + \val_{v(i+1)}' =
\val_{u(i)} + \val_{v(i)}' + \update_{v(i+1)}' =
\val_i'' + \update_{u(i+1)}$ for all $0 \le i < k+l$.
Hence, the claim holds. \qed
\end{proof}

\begin{proposition}
\label{prop:repetition-lemma}
Let $U,W$ be a partitioning of $\vars$.
Let $d \ge 1$ be a natural number and let $\ppath$ be a pre-path such that $\valueSum(\ppath)(\var) \ge 0$ for all $\var \in U$.
If $\ppath$ can be executed from some valuation $\val$,
then $\ppath^d$ can be executed from valuation $\val_d$ with $\val_d(\var) = \val(\var)$ for $\var \in U$, and $\val_d(\var) = d\val(\var)$, for $\var \in W$.
\end{proposition}
\begin{proof}
Let $\val$ be a valuation from which $\ppath$ can be executed and let $\val'$ be the valuation reached by $\ppath$ from $\val$.
Because of $\val + \valueSum(\ppath) = \val' \ge 0$, we get that $\valueSum(\ppath) \ge -\val$ (*).

We now prove the claim by induction on $d \ge 1$.
Clearly the claim holds for $d=1$.
We consider some $d > 1$.
Because of $\val_d \ge \val$ we have that $\ppath$ can be executed from valuation $\val_d$.
Let $\val_d'$ be the valuation reached from $\val_d$ by executing $\ppath$.
By (*) we have $\val_d'(\var) \ge \val_d(\var)-\val(\var)=
\val_{d-1}(\var)$ for all $\var \in W$.
Because of $\valueSum(\ppath)(\var) \ge 0$ for all $\var \in U$ we have $\val_d'(\var) \ge \val_d(\var) = \val(\var) = \val_{d-1}(\var)$ for all $\var \in U$.
Hence, the claim follows from the induction assumption. \qed
\end{proof}

\subsection{Step I}
\label{subsec:step-one}

The properties stated in the two lemmata below are needed for the construction of the pre-paths $\ppathLayer_\layer$ along the tree $\ctree$.
These properties are direct consequences of
constraint system~($\conSysI$), and are the key ingredient for the lower bound proof.

\begin{lemma}
\label{lem:constraint-system-property-one}
Let $0 \le i < \layer$ be some layers.
For every $\node \in \layerF(i)$ and variable $\var \in \vars$ with $\vExp(\var) \le \layer-i$ we have $$\sum_{\node' \in \layerF(\layer), \node' \text{ is descendent of } \node} \valueSum(\cycle(\node'))(\var) \ge 0.$$
\end{lemma}
\begin{proof}
We first start with a statement that will be helpful to prove the claim.
Let $\var \in \vars$ be a variable with $\vExp(\var) \le \layer$.
Let $\counters$ be the optimal solution to constraint system~(\conSysI) during loop iteration $\layer$ of Algorithm~\ref{alg:algorithm}.
We consider some node $\node_\circ \in \layerF(\layer - \vExp(\var))$ in layer $\layer - \vExp(\var)$ of $\ctree$.
Because $\counters$ is a solution to~(\conSysI) we have that
\begin{displaymath}
       \sum_{\transition \in \transitions(\vassF(\node_\circ)) \text{ and } \transition \in \transitions(\vassF(\node)) \text{ for some } \node \in \layerF(\layer)} \updates(\var,\transition) \cdot \counters(\transition) \ge 0 \ (*).
\end{displaymath}
We recall that we have
$\valueSum(\cycle(\node)) = \sum_{\transition \in \transitions(\vassF(\node))} \updates(\transition) \cdot \counters(\transition)$ for all $\node \in \layerF(\layer)$ by the definition of the cycles $\cycle(\node)$.
With (*) we get
\begin{displaymath}
    \sum_{\node \in \layerF(\layer), \node \text{ is descendent of } \node_\circ} \valueSum(\cycle(\node))(\var) \ge 0 \ (\#).
\end{displaymath}

We are now ready to prove the claim.
We now consider some node $\node \in \layerF(i)$ and some variable $\var \in \vars$ with $\vExp(\var) \le \layer-i$.
We have
\begin{multline*}
\sum_{\node' \in \layerF(\layer), \node' \text{ is descendent of } \node} \valueSum(\cycle(\node'))(\var)
= \\
\sum_{\node_\circ \in \layerF(\layer-\vExp(\var)), \node_\circ \text{ is descendent of } \node}
\left( \sum_{\node' \in \layerF(\layer), \node' \text{ is descendent of } \node_\circ} \valueSum(\cycle(\node'))(\var) \right)\\
 \ge 0,
\end{multline*}
where we have the last inequality from (\#). \qed
\end{proof}

\begin{lemma}
\label{lem:constraint-system-property-two}
Let $\layer \ge 1$ be a layer.
Then we have
$$\sum_{\node \in \layerF(\layer)} \valueSum(\cycle(\node))(\var) \ge 1$$
for variables $\var \in \vars$ with $\vExp(\var) > \layer$.
\end{lemma}
\begin{proof}
Let $\var \in \vars$ be some variable with $\vExp(\var) > \layer$.
Let $\counters$ and $\rankCoeff, \offsets$ be the optimal solutions to constraint systems~(\conSysI) and~($\conSysII$) during loop iteration $\layer$ of Algorithm~\ref{alg:algorithm}.
Because $\counters$ is a solution to~(\conSysI) we have that
$$\sum_{\transition \in \transitions(\vassF(\node)) \text{ for some } \node \in \layerF(\layer)} \updates(\var,\transition) \cdot \counters(\transition) \ge 0. (*)$$
Because of $\vExp(\var) > \layer$ we must have $\vExp(\var) = \nobound$ during iteration $\layer$ of Algorithm~\ref{alg:algorithm}.
We now observe that we must have $\rankCoeff(\var,\Root) = 0$ (otherwise Algorithm~\ref{alg:algorithm} would set $\vExp(\var) := \layer$ during loop iteration $\layer$, contradicting the assumption $\vExp(\var) > \layer$).
From the dichotomy stated in Lemma~\ref{lem:optimization-duality} we then get that the inequality (*) must be strict.
Now the claim follows because of
$\valueSum(\cycle(\node)) = \sum_{\transition \in \transitions(\vassF(\node))} \updates(\transition) \cdot \counters(\transition)$ by definition of the cycles $\cycle(\node)$. \qed
\end{proof}

For the construction of the pre-paths $\ppathLayer_\layer$, we need the following convention:
For every layer $\layer \ge 0$ and every node $\node \in \layerF(\layer)$, we consider the cyclic path $\cycle(\node)$, and fix once and for all a decomposition $\cycle(\node) = \paath_0\paath_1\paath_2\cdots \paath_\NoC$ and an ordering of the children $\node_1,\ldots,\node_\NoC$ of $\node$ such that each $\paath_j$ has the same start state as $\cycle(\node_j)$ (e.g., we can order the children $\node_1,\ldots,\node_\NoC$ of $\node$ according to the first appearance of the start state of the path $\cycle(\node_j)$ in the path $\cycle(\node)$).

We will now define the pre-paths $\ppathLayer_\layer$ along the structure of the tree $\ctree$.
In order to do so, we will define pre-paths $\ppath_\layer(\node)$ for all layers $\layer \ge 1$ and nodes $\node \in \layerF(i)$ with $0 \le i \le \layer$.
We will then set $\ppathLayer_\layer = \ppath_\layer(\Root)^N$ for all layers $\layer \ge 1$.

We define the pre-paths $\ppath_\layer(\node)$ inductively, starting from $i = \layer$ downto $i = 0$.
For $\node \in \layerF(\layer)$ we set $\ppath_\layer(\node) = \cycle(\node)$.
We now consider some $\node \in \layerF(i)$ with $0 \le i < \layer$.
Let $\cycle(\node) = \paath_0\paath_1\paath_2\cdots \paath_\NoC$ be the fixed decomposition and $\node_1,\ldots,\node_\NoC$ the corresponding ordering of the children of $\node$ such that each $\paath_j$ has the same start state as $\cycle(\node_j)$.
We set $\ppath_\layer(\node) = \ppath_\layer(\node_1)^N \ppath_\layer(\node_2)^N \cdots \ppath_\layer(\node_\NoC)^N$.
We now set $\ppathLayer_\layer = \ppath_\layer(\Root)^N$ for all layers $\layer \ge 1$.

We show the following properties of the pre-paths $\ppath_\layer(\node)$:
\begin{lemma}
\label{lem:pre-paths-stages-properties}
 For all $\layer \ge 1$ and nodes $\node \in \layerF(i)$ with $0 \le i \le \layer$ we have:
\begin{enumerate}[label=\arabic*),ref=\arabic*]
  \item For every $\node' \in \layerF(\layer)$ that is a descendant of $\node$ and every transition   $\transition \in \transitions(\vassF(\node'))$ we have $\instances(\ppath_\layer(\node),\transition) \ge N^{\layer-i}$ \label{enum:prop-zero}.
  \item $\valueSum(\ppath_\layer(\node)) = N^{\layer-i} \sum_{\node' \in \layerF(\layer), \node' \text{ is descendent of } \node} \valueSum(\cycle(\node'))$ \label{enum:prop-one}.
  \item $\valueSum(\ppath_\layer(\node))(\var) \ge 0$ for every $\var \in \vars$ with $\vExp(\var) \le \layer-i$.
      \label{enum:prop-new}
  \item $\ppath_\layer(\node)$ is executable from some valuation $\val$ with
      \begin{enumerate}[label=\alph*), ref=\theenumi{}\alph*]
           \item $\val(\var) \in O(N^{\vExp(\var)})$ for $\var \in \vars$ with $\vExp(\var) \le \layer-i$, and \label{enum:prop-two-a}
           \item $\val(\var) \in O(N^{\layer-i})$ for $\var \in \vars$ with  $\vExp(\var) \ge \layer-i+1$. \label{enum:prop-two-b}
      \end{enumerate} \label{enum:prop-two}
\end{enumerate}
\end{lemma}
\begin{proof}
We prove the properties for nodes $\node \in \layerF(i)$ by induction, starting from $i = \layer$ downto $i = 0$.

Assume $i = \layer$:
We fix some $\node \in \layerF(\layer)$.
We have that $\ppath_\layer(\node) = \cycle(\node)$ contains at least one instance of each transition $\transition \in \transitions(\vassF(\node))$.
Hence, Property~\ref{enum:prop-zero}) holds.
We have $\valueSum(\ppath_\layer(\node)) = \valueSum(\cycle(\node))$ because of $\ppath_\layer(\node) = \cycle(\node)$.
Hence, Property~\ref{enum:prop-one}) holds.
We note that $\vExp(\var) > \layer-i = 0$ for all variables $\var \in \vars$ (*).
Hence, Property~\ref{enum:prop-new}) trivially holds.
Because of (*), we only have to establish~\ref{enum:prop-two-b}).
We observe that there is a $c > 0$ such that  $\ppath_\layer(\node) = \cycle(\node)$ can be executed from valuation $c \cdot \oneVec$.
Hence, Property~\ref{enum:prop-two}) holds.

Assume $i < \layer$:
We fix some $\node \in \layerF(i)$.
Let $\cycle(\node) = \paath_0\paath_1\paath_2\cdots \paath_\NoC$ be the fixed decomposition and $\node_1,\ldots,\node_\NoC$ the corresponding ordering of the children of $\node$ such that each $\paath_j$ has the same start state as $\cycle(\node_j)$.

We show Property~\ref{enum:prop-zero}):
By induction assumption we have for every $1 \le j \le \NoC$ and every node $\node' \in \layerF(\layer)$ that  is a descendant of $\node_j$ that $\ppath_\layer(\node_j)$ contains at least $N^{\layer-i-1}$ instances of every $\transition \in \transitions(\vassF(\node'))$.
The claim then follows because $\ppath_\layer(\node)$ contains $N$ copies of each $\ppath_\layer(\node_j)$.

We show Property~\ref{enum:prop-one}):
From the induction assumption we get
$\valueSum(\ppath_\layer(\node_j)) = N^{\layer-i-1} \sum_{\node' \in \layerF(\layer), \node' \text{ is descendent of } \node_j} \cycle(\node')$
for all $1 \le j \le \NoC$.
Hence, we have
$\valueSum(\ppath_\layer(\node)) =
N \sum_{1 \le j \le \NoC}  \valueSum(\ppath_\layer(\node_j)) =$\\
$N \sum_{1 \le j \le \NoC} N^{\layer-i-1}  \sum_{\node' \in \layerF(\layer), \node' \text{ is descendent of } \node_j} \valueSum(\cycle(\node')) =$\\
$N^{\layer-i} \sum_{1 \le j \le \NoC} \sum_{\node' \in \layerF(\layer), \node' \text{ is descendent of } \node_j} \valueSum(\cycle(\node')) =$\\
$N^{\layer-i} \sum_{\node' \in \layerF(\layer), \node' \text{ is descendent of } \node} \valueSum(\cycle(\node'))$,
where the last equality holds because every $\node'$ that is a descendent of $\node$ is also a descendent of some $\node_j$.

We show Property~\ref{enum:prop-new}):
The property is a direct consequence of Property~\ref{enum:prop-one}) and Lemma~\ref{lem:constraint-system-property-one}.

We show Property~\ref{enum:prop-two}):
By induction assumption we have that each $\ppath_\layer(\node_j)$ can be executed from some valuation $\val_j$ with $\val_j(\var) \in O(N^{\vExp(\var)})$, for $\var \in \vars$ with $\vExp(\var) \le \layer-i-1$, and $\val_j(\var) \in O(N^{\layer-i-1})$, otherwise.
By Property~\ref{enum:prop-new}) we have for each $\ppath_\layer(\node_j)$ that $\valueSum(\ppath_\layer(\node_j))(\var) \ge 0$ for every $\var \in \vars$ with $\vExp(\var) \le \layer-i-1$.
With Proposition~\ref{prop:repetition-lemma} we get that each $\ppath_\layer(\node_j)^N$ can be executed from some valuation $\val_j$ with $\val_j(\var) \in O(N^{\vExp(\var)})$, for $\var \in \vars$ with $\vExp(\var) \le \layer-i-1$, and $\val_j(\var) \in O(N^{\layer-i})$, otherwise.
With Proposition~\ref{prop:merge-proposition} we get that $\ppath_\layer(\node) = \ppath_\layer(\node_1)^N \ppath_\layer(\node_2)^N \cdots \ppath_\layer(\node_\NoC)^N$
can be executed from some valuation $\val$ with $\val(\var) \in O(N^{\vExp(\var)})$, for $\var \in \vars$ with $\vExp(\var)\le \layer-i-1$, and $\val(\var) \in O(N^{\layer-i})$, otherwise.
We note that the last statement implies  Property~\ref{enum:prop-two}).
\qed
\end{proof}

We show the following properties of the pre-paths $\ppathLayer_\layer$:
\begin{lemma}
\label{lem:pre-paths-complete-properties}
For all $\layer \ge 1$ we have:
\begin{enumerate}[label=\arabic*),ref=\arabic*]
  \item $\instances(\ppathLayer_\layer,\transition) \ge N^{\layer+1}$ for all transitions $\transition \in \bigcup_{\node \in \layerF(\layer)} \transitions(\vassF(\node))$. \label{enum:prop-complete-one}
  \item $\valueSum(\ppathLayer_\layer) = N^{\layer+1} \sum_{\node \in \layerF(\layer)} \valueSum(\cycle(\node'))$. \label{enum:prop-complete-two}
  \item $\valueSum(\ppathLayer_\layer)(\var) \ge 0$ for every $\var \in \vars$ with $\vExp(\var) \le \layer$.
      \label{enum:prop-complete-three}
  \item $\valueSum(\ppathLayer_\layer)(\var) \ge N^{\layer+1}$ for every $\var \in \vars$ with $\vExp(\var) \ge \layer+1$. \label{enum:prop-complete-four}
  \item $\ppathLayer_\layer$ is executable from some valuation $\val$ with
      \begin{enumerate}[label=\alph*), ref=\theenumi{}\alph*]
           \item $\val(\var) \in O(N^{\vExp(\var)})$ for $\var \in \vars$ with $\vExp(\var) \le \layer$, and \label{enum:prop-complete-five-a}
           \item $\val(\var) \in O(N^\layer)$ for $\var \in \vars$ with  $\vExp(\var) \ge \layer+1$. \label{enum:prop-complete-five-b}
      \end{enumerate} \label{enum:prop-complete-five}
\end{enumerate}
\end{lemma}
\begin{proof}
Let $\layer \ge 1$ be some layer.
We consider the pre-path $\ppath_\layer(\Root)$.
Using the fact that every node in layer $\layer$ is a descendant of the root node $\Root$,
we get from Property~\ref{enum:prop-complete-two}) in Lemma~\ref{lem:pre-paths-stages-properties} that
$\valueSum(\ppath_\layer(\Root)) = N^\layer \sum_{\node \in \layerF(\layer)} \valueSum(\cycle(\node))$ (*).
From Lemma~\ref{lem:constraint-system-property-two} we now get that $\valueSum(\ppath_\layer(\Root))(\var) \ge N^\layer$ for every $\var \in \vars$ with $\vExp(\var) \ge \layer+1$ (\#).
From Property~\ref{enum:prop-complete-three}) in Lemma~\ref{lem:pre-paths-stages-properties} and (\#) we get in particular that $\valueSum(\ppath_\layer(\Root))(\var) \ge 0$ for all $\var \in \vars$ (+).

We recall that we defined $\ppathLayer_\layer = \ppath_\layer(\Root)^N$, i.e., that $\ppathLayer_\layer$ consists of $N$ copies of $\ppath_\layer(\Root)$.
Then, Properties~\ref{enum:prop-complete-one}) and~\ref{enum:prop-complete-three}) directly follow from the corresponding properties in Lemma~\ref{lem:pre-paths-stages-properties}.
Properties~\ref{enum:prop-complete-two}) and~\ref{enum:prop-complete-four}) follow from (*) and (\#).
Property~\ref{enum:prop-complete-five})
follows from Property~\ref{enum:prop-two} in Lemma~\ref{lem:pre-paths-stages-properties} using (+) and
Proposition~\ref{prop:repetition-lemma}.
\qed
\end{proof}

\subsection{Step II}

By Property~\ref{enum:prop-complete-five}) of Lemma~\ref{lem:pre-paths-complete-properties} we can choose a sufficiently large $k > 0$ such that every pre-path $\ppathLayer_\layer$, for $\layer \ge 1$, is executable from valuation $\val_\layer$ with
\begin{enumerate}[label=\alph*),ref=\alph*]
   \item $\val_\layer(\var) = k N^{\vExp(\var)}$ for $\var \in \vars$ with $\vExp(\var) \le \layer$, and
   \item $\val_\layer(\var) = k N^\layer$ for $\var \in \vars$ with $\vExp(\var) \ge \layer+1$.
\end{enumerate}
For every $\layer \ge 0$, we now define the pre-path
$\ppathComplete_\layer = \ppathLayer_0^k \ppathLayer_1^k \ppathLayer_2^k \cdots \ppathLayer_\layer^k$ (the concatenation of $k$ copies of the pre-paths $\ppathLayer_i$ for all $0 \le i \le \layer$),
where we set $\ppathLayer_0 = \cycle(\Root)^N$.

The following lemma shows that the pre-path $\ppathComplete_{\layer_{\max}}$, where $\layer_{\max}$ is the maximal layer of $\ctree$, would be sufficient to conclude the lower bound proof except that we will need to extend the proof from pre-paths to proper paths.

\begin{lemma}
\label{lem:ppathComplete-meets-lower-bounds}
For every $\layer \ge 0$, the pre-path $\ppathComplete_\layer$ can be executed from a valuation $\val$ and reaches a valuation $\val'$ with
\begin{enumerate}[label=\arabic*), ref=\arabic*]
  \item $\norm{\val} \in O(N)$,
  \item $\val'(\var) \ge k N^{\vExp(\var)}$ for all $\var \in \vars$ with $\vExp(\var) \le \layer$, and
  \item $\val'(\var) \ge k N^{\layer+1}$ for all $\var \in \vars$ with $\vExp(\var) \ge \layer+1$.
\end{enumerate}
\end{lemma}
\begin{proof}
The proof proceeds by induction on $\layer$.
We consider $\layer = 0$.
We have $\ppathComplete_0 = \ppathLayer_0^k = \cycle(\Root)^{kN}$.
Clearly, there is a $c > 0$ such that $\cycle(\Root)^{kN}$ can be executed from valuation $cN \cdot \oneVec$ such that $cN \cdot \oneVec + kN \cdot \valueSum(\cycle(\Root)) \ge kN \cdot \oneVec$.
This establishes the base case.

We consider some $\layer \ge 1$ and assume the induction assumption for $\ppathComplete_{\layer-1}$.
We observe that $\ppathComplete_\layer = \ppathComplete_{\layer-1} \ppathLayer_\layer^k$.
By Property~\ref{enum:prop-complete-two}) of Lemma~\ref{lem:pre-paths-complete-properties} we have $\valueSum(\ppathLayer_\layer) = N^{\layer+1} \sum_{\node \in \layerF(\layer)} \valueSum(\cycle(\node))$.
By Properties~\ref{enum:prop-complete-three}) and~\ref{enum:prop-complete-four}) of Lemma~\ref{lem:pre-paths-complete-properties} we have
$\ppathLayer_\layer(\var) \ge 0$, for every $\var \in \vars$ with $\vExp(\var) \le \layer$, and $\valueSum(\ppathLayer_\layer)(\var) \ge N^{\layer+1}$ for every $\var \in \vars$ with $\vExp(\var) \ge \layer+1$ (*).
Hence, we have
$\valueSum(\ppathLayer_\layer^k)(\var) \ge 0$, for variables $\var \in \vars$ with $\vExp(\var) \le \layer$, and $\valueSum(\ppathLayer_\layer^k)(\var) \ge k N^{\layer + 1}$, for $\var \in \vars$ with $\vExp(\var) \ge \layer+1$.
By the definition of $k$ we have that $\ppathLayer_\layer$ can be executed from valuation $\val(\var) = k N^{\vExp(\var)}$, for $\var \in \vars$ with $\vExp(\var) \le \layer$, and $\val(\var) = k N^\layer$, for $\var \in \vars$ with $\vExp(\var) \ge \layer+1$.
By Proposition~\ref{prop:repetition-lemma} and (*), $\ppathLayer_\layer^k$ can be executed from valuation $\val(\var) = k N^{\vExp(\var)}$, for $\var \in \vars$ with $\vExp(\var) \le \layer$, and $\val(\var) = k N^\layer$, for $\var \in \vars$ with $\vExp(\var) \ge \layer+1$.
Because of $\ppathComplete_\layer = \ppathComplete_{\layer-1} \ppathLayer_\layer^k$, the claim now follows from the induction assumption. \qed
\end{proof}

\subsection{Step III}

In order to extend the proof from pre-paths to paths we make use of the concept of shuffling.
For all $\layer \ge 0$, we will define paths $\witComplete_\layer$ that can be obtained by shuffling the pre-paths $\ppathComplete_0, \ppathComplete_1, \ldots, \ppathComplete_\layer$.

We will first define paths $\witLayer_\layer$ analogously to the pre-paths $\ppathLayer_\layer$.
We will then set $\witComplete_\layer = \witLayer_0^k\witLayer_1^k\cdots\witLayer_\layer^k$ for all $\layer \ge 0$, where $k$ is the constant from Step II.

We will define the paths $\witLayer_\layer$ along the structure of $\ctree$.
In order to do so, we will define paths $\wit_\layer(\node)$ for all layers $\layer \ge 0$ and nodes $\node \in \layerF(i)$ with $0 \le i \le \layer$.
We will then set $\witLayer_\layer = \wit_\layer(\Root)^N$ for all $\layer \ge 0$.

We define the paths $\wit_\layer(\node)$ inductively, starting from $i = \layer$ downto $i = 0$.
For $\node \in \layerF(\layer)$ we set $\wit_\layer(\node) = \cycle(\node)$.
We consider some $\node \in \layerF(i)$ with $0 \le i < \layer$.
Let $\cycle(\node) = \paath_0\paath_1\paath_2\cdots \paath_\NoC$ be the fixed decomposition and $\node_1,\ldots,\node_\NoC$ the corresponding ordering of the children of $\node$ such that each $\paath_j$ has the same start state as $\cycle(\node_j)$.
We set $\wit_\layer(\node) = \paath_0\wit_\layer(\node_1)^N\paath_1 \wit_\layer(\node_2)^N \paath_2 \cdots \wit_\layer(\node_\NoC)^N\paath_\NoC$.
We finally set $\witLayer_\layer = \wit_\layer(\Root)^N$ for all $\layer \ge 0$.

\begin{lemma}
\label{lem:merge-per-layer}
The pre-path $\ppath_\layer(\node)$ can be shuffled with the path $\wit_{\layer-1}(\node)$ to obtain the path $\wit_\layer(\node)$ for every $\node \in \layerF(i)$ with $0 \le i < \layer$.
Further, the pre-path $\ppathLayer_\layer$ can be shuffled with the path $\witLayer_{\layer-1}$ to obtain the path $\witLayer_\layer$ for every $\layer \ge 1$.
\end{lemma}
\begin{proof}
For the first claim, we consider some $\node \in \layerF(i)$ with $0 \le i < \layer$.
Let $\cycle(\node) = \paath_0\paath_1\paath_2\cdots \paath_\NoC$ be the fixed decomposition and $\node_1,\ldots,\node_\NoC$ the corresponding ordering of the children of $\node$ such that each $\paath_j$ has the same start state as $\cycle(\node_j)$.
We proceed by induction, starting from $i = \layer -1$ downto $i = 0$.

Assume $i = \layer - 1$:
By definition, we have
$\ppath_\layer(\node) = \cycle(\node_1)^N \cycle(\node_2)^N \cdots \cycle(\node_\NoC)^N$, $\wit_{\layer-1}(\node) = \cycle(\node) = \paath_0\paath_1\paath_2\cdots \paath_\NoC$ and
$\wit_\layer(\node) = \paath_0\cycle(\node_1)^N\paath_1 \cycle(\node_2)^N \paath_2 \cdots \cycle(\node_\NoC)^N \paath_\NoC$.
Clearly, we can shuffle $\cycle(\node_1)^N \cycle(\node_2)^N \cdots \cycle(\node_\NoC)^N$ with $\paath_0\paath_1\paath_2\cdots \paath_\NoC$ in order to obtain $\paath_0\cycle(\node_1)^N\paath_1 \cycle(\node_2)^N \paath_2 \cdots \cycle(\node_\NoC)^N\paath_\NoC$.
Hence, the claim holds.

Assume $i < \layer - 1$:
By definition, we have
$\ppath_\layer(\node) = \ppath_\layer(\node_1)^N \ppath_\layer(\node_2)^N \cdots \ppath_\layer(\node_\NoC)^N$,
$\wit_{\layer-1}(\node) = \paath_0\wit_{\layer-1}(\node_1)^N\paath_1 \wit_{\layer-1}(\node_2)^N \paath_2 \cdots \wit_{\layer-1}(\node_\NoC)^N \paath_\NoC$ and
$\wit_\layer(\node) =$\\
$\paath_0\wit_\layer(\node_1)^N\paath_1 \wit_\layer(\node_2)^N \paath_2 \cdots \wit_\layer(\node_\NoC)^N \paath_\NoC$.
By induction assumption, we have that each $\ppath_\layer(\node_j)$ can be shuffled with $\wit_{\layer-1}(\node_j)$ to obtain $\wit_\layer(\node_j)$.
Clearly, then also $\ppath_\layer(\node_j)^N$ can be shuffled with $\wit_{\layer-1}(\node_j)^N$ to obtain $\wit_\layer(\node_j)^N$.
We then get that $\ppath_\layer(\node)$ can be shuffled with $\wit_{\layer-1}(\node)$ to obtain the path $\wit_\layer(\node)$.

For the second claim we consider some $\layer \ge 1$.
By the above, we have that $\ppath_\layer(\Root)$ can be shuffled with $\wit_{\layer-1}(\Root)$ to obtain the path $\wit_\layer(\Root)$.
Clearly, then also $\ppathLayer_\layer = \ppath_\layer(\Root)^N$ can be shuffled with $\witLayer_{\layer-1}(\Root)^N$ to obtain $\witLayer_\layer = \wit_\layer(\Root)^N$.
\qed
\end{proof}

We now set $\witComplete_\layer = \witLayer_0^k\witLayer_1^k\cdots\witLayer_\layer^k$ for all $\layer \ge 0$, where $k$ is the constant from Step II.
The main property of the paths $\witComplete_\layer$ is stated in the next lemma:

\begin{lemma}
For every $\layer \ge 0$, $\witComplete_\layer$ can be obtained by shuffling the pre-paths $\ppathComplete_0,\ppathComplete_1,\ldots,\ppathComplete_\layer$.
\end{lemma}
\begin{proof}
The proof is by induction on $\layer \ge 0$.
We first consider $\layer = 0$.
Then, we have $\ppathComplete_0 = \witComplete_0 = \cycle(\node)^{kN}$ and the claim trivially holds.

We assume $\layer \ge 1$.
By induction assumption
we can obtain $\witComplete_{\layer-1}$ by shuffling
$\ppathComplete_0,\ppathComplete_1,\ldots,\ppathComplete_{\layer-1}$. We will now argue that we can obtain
$\witComplete_\layer = \witLayer_0^k\witLayer_1^k\cdots\witLayer_\layer^k$
by shuffling $\ppathComplete_\layer$ with $\witComplete_{\layer-1} = \witLayer_0^k\witLayer_1^k\cdots\witLayer_{\layer-1}^k$
By definition we have $\ppathComplete_\layer = \ppathLayer_0^k \ppathLayer_1^k \cdots \ppathLayer_\layer^k$.
By Lemma~\ref{lem:merge-per-layer}, $\ppathLayer_i$ can be shuffled with $\witLayer_{i-1}$ to obtain $\witLayer_i$ for every $0 < i \le \layer$.
But then we can also shuffle $\ppathLayer_i^k$ with $\witLayer_{i-1}^k$ to obtain $\witLayer_i^k$ for every $0 < i \le \layer$.
Because $\ppathLayer_i^k$ and $\witLayer_{i-1}^k$ appear in the same order in $\ppathLayer_1^k \ppathLayer_2^k \cdots \ppathLayer_\layer^k$ and $\witLayer_0^k\witLayer_1^k\cdots\witLayer_{\layer-1}^k$, we can shuffle these two (pre-)paths and obtain $\witLayer_1^k\cdots\witLayer_\layer^k$.
The claim then follows because of $\witLayer_0 = \ppathLayer_0 = \cycle(\node)^N$. \qed
\end{proof}

We finally state the proof of Theorem~\ref{thm:lower-bound}:
\begin{proof}
Let $\layer_{\max}$ be the maximal layer of $\ctree$.
We will now argue that the path $\witComplete_{\layer_{\max}}$ gives rise to traces with the desired properties.
By Lemma~\ref{lem:ppathComplete-meets-lower-bounds}, each pre-path $\ppathComplete_i$, for $0 \le i \le \layer_{\max}$, can be executed from a valuation $\val$ with $\norm{\val} \in O(N)$, and $\ppathComplete_{\layer_{\max}}$ reaches a valuation $\val'$ with $\val'(\var) \ge N^{\vExp(\var)}$ for all $\var \in \vars$.
From Proposition~\ref{prop:merge-proposition} we get that the path $\witComplete_{\layer_{\max}}$,
which can be obtained by shuffling $\ppathComplete_0,\ppathComplete_1,\ldots,\ppathComplete_{\layer_{\max}}$,
can be executed from a valuation $\val$
with $\norm{\val} \in O(N)$ and reaches a valuation $\val'$ with $\val'(\var) \ge N^{\vExp(\var)}$ for all $\var \in \vars$.
We consider a transition $\transition$ with $\tExp(\transition) = i$ for some $1 \le i \le {\layer_{\max}}$.
There is a node $\node \in \layerF(i-1)$ such that $\transition \in \transitions(\vassF(\node))$.
By Property~\ref{enum:prop-zero}) of Lemma~\ref{lem:pre-paths-stages-properties} we have that $\ppathLayer_{i-1}$ contains at least $N^i$ instances of $\transition$.
Hence, $\ppathComplete_{i-1} = \ppathLayer_0^k \ppathLayer_1^k \ppathLayer_2^k \cdots \ppathLayer_{i-1}^k$ also contains at least $N^i$ instances of $\transition$.
Thus, $\witComplete_{\layer_{\max}}$,
which has been obtained from shuffling $\ppathComplete_0,\ppathComplete_1,\ldots,\ppathComplete_{\layer_{\max}}$,
also contains at least $N^i$ instances of $\transition$.
\qed
\end{proof}

\section{Proof of Proposition~\ref{prop:pot-estimation}}

\begin{proof}
  Let $L$ be a lower triangular matrix with $\pot(Z) = \sum_{\text{column } \rankCoeff \text{ of } L} \pot(\rankCoeff)$.
  Because $L$ is a lower triangular matrix, we have that for every variable $\var \in \vars$ there is at most one column vector $\rankCoeff$
  such that $\rankCoeff(\var) \neq 0$ and $\rankCoeff(\var') = 0$ for all variables $\var'$    such the row associated to variable $\var'$ is above the row associated to variable $\var$.
  Now the claim follows because we have $\pot(\rankCoeff) = \vExp(\var)$ for every column vector $\rankCoeff$ of $L$ and every variable $\var$ such that $\rankCoeff(\var) \neq 0$ and $\rankCoeff(\var') = 0$ for all variables $\var'$ such the row associated to variable $\var'$ is above the row associated to variable $\var$.
  \qed
\end{proof}

\section{Proof of Lemma~\ref{lem:entering-property}}

\begin{proof}
The proof proceeds by induction on the layer $\layer$ of $\node \in \layerF(\layer)$.
Clearly, the claim holds for the root $\Root \in \layerF(0)$.

We now assume that the claim holds for $\layer \ge 0$ and prove the claim for $\layer+1$.
We fix some node $\node \in \layerF(\layer+1)$.
Let $\node_p \in \layerF(\layer)$ be the parent of $\node$.
We will apply the Bound Proof Principle from Section~\ref{sec:upper-bound} in order to show that
every trace $\trace$ with $\initBound(\trace) \le N$ enters $\vassF(\node)$ at most $O(N^{\pot(\node)})$ times.

We first need to prepare for the definition of the witness function.
By Property~\ref{enum:inv-prop-one}) of Proposition~\ref{prop:inv-set-properties} we have
$\rankCoeff_{\node_p}^T\update + \offsets_\layer(\state_2) - \offsets_\layer(\state_1) \le 0$ for all $\state_1 \xrightarrow{\update} \state_2 \in \transitions(\vassF(\node_p))$ (*).
Further, by Property~\ref{enum:inv-prop-two}) of Proposition~\ref{prop:inv-set-properties} we have for every strict ancestor $\node' \in \layerF(i)$ of $\node_p$ that
$\rankCoeff_{\node'}^T\update + \offsets_i(\state_2) - \offsets_i(\state_1) = 0$ for all $\state_1 \xrightarrow{\update} \state_2 \in \transitions(\vassF(\node_p))$ (\#).
Now we can choose coefficients $\lambda_{\node'} \in \mathbb{Z}$, for every strict ancestor $\node'$ of $\node_p$, and a coefficient $\lambda_{\node_p} \in \mathbb{Z}$ with $\lambda_{\node_p} > 0$ such that the vector $\rankCoeff^\circ = \lambda_{\node_p} \rankCoeff_{\node_p} + \sum_{\text{strict ancestor } \node' \text{ of } \node_p} \lambda_{\node'} \rankCoeff_{\node'}$ satisfies the equation $\pot(\rankCoeff^\circ) + \pot(\node_p) = \pot(\node)$
(the coefficients $\lambda_{\node'}$ and $\lambda_{\node_p}$ can be chosen according to the elementary column operations in the definition of $\pot(\node)$).
Because of Property~\ref{enum:inv-prop-six} of Proposition~\ref{prop:inv-set-properties} we can in fact choose the coefficients $\lambda_{\node'}$ such that $\rankCoeff^\circ = \lambda_{\node_p} \rankCoeff_{\node_p} + \sum_{\text{strict ancestor } \node' \text{ of } \node_p} \lambda_{\node'} \rankCoeff_{\node'} \ge 0$.
We now consider the vector $\offsets^\circ = \lambda_{\node_p} \offsets_{\node_p} + \sum_{0 \le i < \layer} \lambda_{\node'} \offsets_i$.
By (*) and (\#) we get that
${\rankCoeff^\circ}^T\update + \offsets^\circ(\state_2) - \offsets^\circ(\state_1) \le 0$ for all $\state_1 \xrightarrow{\update} \state_2 \in \transitions(\vassF(\node_p))$ (+);
we note that this inequality is strict for all transitions $\transition$ for which inequality (*) is strict, i.e., for all $\transition \in \transitions(\vassF(\node_p))$ with  $\tExp(\transition) = \layer+1$.
Further, the inequality (+) remains valid, if we add the vector $c \cdot \oneVec$ to $\offsets^\circ$ for any $c \in \mathbb{Z}$; hence, we can assume $\offsets^\circ \ge 0$.

We now define the witness function $\coeff: \configs(\vass) \rightarrow \nats$ by setting $\coeff(\state,\val) = {\rankCoeff^\circ}^T \val + \offsets^\circ(\state)$, for $\state \in \states(\vassF(\node_p))$, and $\coeff(\state,\val) = 0$, otherwise.
We note that $\coeff$ is well-defined,
i.e., $\coeff(\state,\val) \ge 0$ for all $(\state,\val) \in \configs(\vass)$, because we have shown above that $\rankCoeff^\circ \ge 0$ and $\offsets^\circ \ge 0$.
By (+) we have already established condition 1) of the Bound Proof Principle.
It remains to establish condition 2).

We observe that only transitions that enter $\vassF(\node_p)$ can increase the value of $\coeff$.
Let $\transition = \state_1 \xrightarrow{\update} \state_2$ be such a transition and let $(\state_1,\val_1) \xrightarrow{\update} (\state_2,\val_2)$ be some step in a trace $\trace$ of $\vass$ with $\initBound(\trace) \le N$.
Then, $p(\state_2,\val_2) - p(\state_1,\val_1) = p(\state_2,\val_2) = {\rankCoeff^\circ}^T \val_2 + \offsets^\circ(\state) \in O(N^{\pot(\rankCoeff^\circ)})$.
Hence, there is an increase certificate $\increase_\transition$ with $\increase_\transition(N) \in O(N^{\pot(\rankCoeff^\circ)})$.
We further note that by induction assumption we have that $\trace$ enters $\vassF(\node_p)$ at most $O(N^{\pot(\node_p)})$ times.
Hence, we have $\tbound_\transition(N) \cdot \increase_\transition(N) = O(N^{\pot(\node_p)}) \cdot O(N^{\pot(\rankCoeff^\circ)}) =  O(N^{\pot(\node_p)+\pot(\rankCoeff^\circ)}) = O(N^{\pot(\node)})$.

We are now ready to apply the Bound Proof Principle from Proposition~\ref{prop:bound-proof-principle}.
We observe that $\max_{(\state,\val) \in \configs(\vass), \norm{\val} \le N} \coeff(\state,\val) \in O(N)$ because $\coeff(\state,\val)$ is a linear expression for all $\state \in \states(\vass)$, and we consider valuations $\val$ with $\norm{\val} \le N$.
From the Bound Proof Principle we now get that $\tbound_\transition(N) \in O(N^{\pot(\node)})$ for all transitions $\transition \in \transitions(\vassF(\node_p))$ with $\tExp(\transition) = \layer+1$.
Finally, we observe that every transition $\transition$ that enters $\vassF(\node)$ either also enters $\vassF(\node_p)$ or belongs to $\vassF(\node_p)$ and we have $\tExp(\transition) = \layer+1$.
This concludes the proof.
\qed
\end{proof}

\section{Proof of Lemma~\ref{lem:var-bound-relation}}

\begin{proof}
We observe that $\vExp(\var) = 1$ resp. $\tExp(\transition) = 1$ for the bounds discovered in the first iteration and $\vExp(\var) > 1$ resp. $\tExp(\transition) > 1$ for variable and transition bounds discovered in later iterations of Algorithm~\ref{alg:algorithm}.
We prove the claim by induction on $\layer$.
The claim holds for $\layer = 1$ because we have $\{ \vExp(\var) \mid \var \in \vars, \vExp(\var) < 1\} = \emptyset$ and thus $\varsum(1) = 1$.

We consider some $\layer \ge 1$ and prove the claim for $\layer +1$.
We will apply the Bound Proof Principle from Section~\ref{sec:upper-bound} in order to show that $\vExp(\var) = \layer+1$ resp. $\tExp(\transition) = \layer+1$ implies
$\vExp(\var) \le \varsum(\layer+1)$ resp. $\tExp(\transition) \le \varsum(\layer+1)$.
Concretely, we will show $\vbound_\var(N) \in O(N^{\varsum(\layer+1)})$  resp. $\tbound_\transition(N) \in O(N^{\varsum(\layer+1)})$ for all variables $\var$ with $\vExp(\var) = \layer+1$ and transitions $\transition$ with $\tExp(\transition) = \layer+1$,
and then get $\vExp(\var) \le \varsum(\layer+1)$ and $\tExp(\transition) \le \varsum(\layer+1)$ from $\vbound_N(\var) \in \Omega(N^{\vExp(\var)})$ and $\tbound_N(\transition) \in \Omega(N^{\tExp(\transition)})$ (Corollary~\ref{cor:lower-bound}) .

We first need to prepare for the definition of the witness function.
We define a vector $q \in \mathbb{Z}^\vars$ by setting
$q(\var) = 0$ for variables $\var$ with $\vExp(\var) \neq \layer+1$, and $q(\var) = \rankCoeff_\node(\var)$ for variables $\var$ with $\vExp(\var) = \layer+1$,
where $\node \in \layerF(\layer)$ is arbitrarily chosen.
We now argue that for every $\node \in \layerF(\layer)$ there are vectors $\rankCoeff_\node^\circ \in \mathbb{Z}^\vars$, $\offsets_\node^\circ \in \mathbb{Z}^{\states(\vass)}$ and a coefficient $\lambda_\node > 0$ such that $\pot(\rankCoeff_\node^\circ) + \pot(\node) \le \varsum(\layer+1)$,
$\rankCoeff_\node^\circ \ge 0$, $\offsets_\node^\circ \ge 0$,
and $(\rankCoeff_\node^\circ + \lambda_\node q)^T\update + \offsets_\node^\circ(\state_2) - \offsets_\node^\circ(\state_1) \le 0$ for all $\state_1 \xrightarrow{\update} \state_2 \in \transitions(\vassF(\node))$ (+).
This is done as follows:
We fix some $\node \in \layerF(\layer)$.
By Property~\ref{enum:inv-prop-one}) of Proposition~\ref{prop:inv-set-properties} we have
$\rankCoeff_\node^T\update + \offsets_\layer(\state_2) - \offsets_\layer(\state_1) \le 0$ for all $\state_1 \xrightarrow{\update} \state_2 \in \transitions(\vassF(\node))$ (*);
we note that this inequality is strict for all transitions $\transition \in \transitions(\vassF(\node))$ with $\tExp(\transition) = \layer+1$.
Further, by Property~\ref{enum:inv-prop-two}) of Proposition~\ref{prop:inv-set-properties} we have for every strict ancestor $\node' \in \layerF(i)$ of $\node_p$ that
$\rankCoeff_{\node'}^T\update + \offsets_i(\state_2) - \offsets_i(\state_1) = 0$ for all $\state_1 \xrightarrow{\update} \state_2 \in \transitions(\vassF(\node))$ (\#).
Now we can choose coefficients $\lambda_{\node'} \in \mathbb{Z}$, for every strict ancestor $\node'$ of $\node_p$, and a coefficient $\lambda_\node \in \mathbb{Z}$ with $\lambda_\node > 0$ such that the vector $\rankCoeff_\node^\circ = \lambda_\node (\rankCoeff_\node - q) + \sum_{\text{strict ancestor } \node' \text{ of } \node_p} \lambda_{\node'} \rankCoeff_{\node'}$ satisfies the equation $\pot(\rankCoeff^\circ) + \pot(\node) =
\pot( \{ \rankCoeff_{\node'} \mid \node' \text{ is a strict ancestor } \text{ of } \node \} \cup \{ \rankCoeff_\node - q\} )$.
By Property~\ref{enum:inv-prop-four} of Proposition~\ref{prop:inv-set-properties} we have that the vector $\rankCoeff_\node - q$ has non-zero coefficients only for variables $\var$ with $\vExp(\var) \le \layer$.
Thus, we obtain $\pot(\rankCoeff^\circ) + \pot(\node) =
\pot( \{ \rankCoeff_{\node'} \mid \node' \text{ is a strict ancestor } \text{ of } \node \} \cup \{ \rankCoeff_\node - q\} ) \le \varsum(\layer+1)$ from Proposition~\ref{prop:pot-estimation}.
Because of Property~\ref{enum:inv-prop-six} of Proposition~\ref{prop:inv-set-properties} we can in fact choose the coefficients $\lambda_{\node'}$ such that $\rankCoeff_\node^\circ = \lambda_\node (\rankCoeff_\node - q) + \sum_{\text{strict ancestor } \node' \text{ of } \node_p} \lambda_{\node'} \rankCoeff_{\node'} \ge 0$.
We now consider the vector $\offsets_\node^\circ = \lambda_\node \offsets_\node + \sum_{0 \le i < \layer} \lambda_{\node'} \offsets_i$.
By (*) and (\#) we get that
$(\rankCoeff_\node^\circ + \lambda_\node q)^T\update + \offsets_\node^\circ(\state_2) - \offsets_\node^\circ(\state_1) \le 0$ for all $\state_1 \xrightarrow{\update} \state_2 \in \transitions(\vassF(\node))$;
we note that this inequality is strict for all transitions $\transition$ for which inequality (*) is strict, i.e., for all $\transition \in \transitions(\vassF(\node))$ with $\tExp(\transition) = \layer+1$.
Further, the inequality remains valid, if we add the vector $c \cdot \oneVec$ to $\offsets_\node^\circ$ for any $c \in \mathbb{Z}$; hence, we can assume $\offsets_\node^\circ \ge 0$.

We are now ready to define the witness function $\coeff: \configs(\vass) \rightarrow \nats$.
Let $\lambda > 0$ be the least common multiple of the coefficients $\lambda_\node$ for $\node \in \layerF(\layer)$.
We define $\coeff$ by setting $\coeff(\state,\val) = (\frac{\lambda}{\lambda_\node}\rankCoeff_\node^\circ+\lambda q)^T \val + \frac{\lambda}{\lambda_\node}\offsets_\node^\circ(\state)$, if there is a node $\node \in \layerF(\layer)$ with $\state \in \states(\vassF(\node))$, and $\coeff(\state,\val) = \lambda q^T\val$, otherwise.
We note that $\coeff$ is well-defined,
i.e.,  $\coeff(\state,\val) \ge 0$ for all $(\state,\val) \in \configs(\vass)$, because of $\lambda > 0$, $q \ge 0$ and because of $\rankCoeff_\node^\circ \ge 0$, $ \offsets_\node^\circ \ge 0$ for all $\node \in \layerF(\layer)$.
Let $\unboundedTransitions = \bigcup_{\node \in \layerF(\layer)} \transitions(\vassF(\node))$ be the transitions associated to the nodes in layer $\layer$ of the tree.
By (+) we have already established condition 1) of the Bound Proof Principle for all $\transition \in \unboundedTransitions$.
It remains to establish condition 2).
We will argue that for every $\transition \in \transitions(\vass) \setminus \unboundedTransitions$ we can define increase certificates $\increase_\transition(N)$ such that $\tbound_\transition(N) \cdot \increase_\transition(N) \le O(N^{\varsum(\layer+1)})$.

Let $\state_1 \xrightarrow{\update} \state_2 = \transition \in \transitions(\vass) \setminus \unboundedTransitions$ be a transition and let $(\state_1,\val_1) \xrightarrow{\update} (\state_2,\val_2)$ be a step in a trace $\trace$ of $\vass$ with $\initBound(\trace) \le N$.
We note that $\val_2 = \val_1 + \update$.
We proceed by a case distinction on whether $\transition$ enters some $\vassF(\node)$ for $\node \in \layerF(\layer)$.

In case $\transition$ does not enter $\vassF(\node)$ for any $\node \in \layerF(\layer)$, we have that there is no node $\node \in \layerF(\layer)$ with $\state_2 \in \states(\vassF(\node))$.
Hence, we have $\coeff(\state_2,\val_2) = q^T \val_2$.
We get that $\coeff(\state_2,\val_2) - \coeff(\state_1,\val_1) \le \lambda q^T \val_2 - \lambda q^T \val_1 = \lambda q^T\update \in O(1)$ and we can set $\increase_\transition(N)$ to a constant function.
Because of $\transition \in \transitions(\vass) \setminus \unboundedTransitions$, we have $\tExp(\transition) = \layer'$
for some layer $\layer' \le \layer$.
By induction assumption we have
$\tExp(\transition) \le \varsum(\layer')$.
Because $\varsum$ is a monotone function, we have $\varsum(\layer') \le \varsum(\layer) \le \varsum(\layer+1)$.
Hence, we get $\tExp(\transition) \le \varsum(\layer+1)$.
With $\tbound_\transition(N) \in O(N^{\tExp(\transition)})$ (by Theorem~\ref{thm:upper-bound}) we now get that $\tbound_\transition(N) \cdot \increase_\transition(N) \le O(N^{\tExp(\transition)}) \cdot O(1) = O(N^{\varsum(\layer+1)})$.

In case $\transition$ does enter some $\vassF(\node)$ with $\node \in \layerF(\layer)$, we have
$\coeff(\state_2,\val_2) - \coeff(\state_1,\val_1) \le
(\frac{\lambda}{\lambda_\node} \rankCoeff_\node^\circ+ \lambda q)^T\val_2 + \frac{\lambda}{\lambda_\node} \offsets_\node^\circ(\state_2) - \lambda q^T\val_1 = \frac{\lambda}{\lambda_\node} {\rankCoeff_\node^\circ}^T\val_2 + \lambda q^T\update + \frac{\lambda}{\lambda_\node} \offsets_\node^\circ(\state_2) \in O(N^{\pot(\rankCoeff_\node^\circ)})$.
With $\tbound_\transition(N) \le O(N^{\pot(\node)})$ (by Lemma~\ref{lem:entering-property}) we get $\tbound_\transition(N) \cdot \increase_\transition(N) \le O(N^{\pot(\node)}) \cdot O(N^{\pot(\rankCoeff_\node^\circ)}) = O(N^{\pot(\node) + \pot(\rankCoeff_\node^\circ)}) \le O(N^{\varsum(\layer+1)})$.

We are now ready to apply the Bound Proof Principle from Proposition~\ref{prop:bound-proof-principle}.
We observe that $\max_{(\state,\val) \in \configs(\vass), \norm{\val} \le N} \coeff(\state,\val) \in O(N)$ because $\coeff(\state,\val)$ is a linear expression for all $\state \in \states(\vass)$, and we consider valuations $\val$ with $\norm{\val} \le N$.
From the Bound Proof Principle we now get that $\tbound_\transition(N) \in O(N^{\varsum(\layer+1)})$ for all transitions $\transition \in \transitions(\vassF(\node))$ with $\tExp(\transition) = \layer+1$.
Next, we argue that we can also deduce the desired variable bounds.
We consider a variable $\var \in \vars$ with $\vExp(\var) = \layer+1$.
For all $(\state,\val) \in \configs(\vass)$ we have $\coeff(\state,\val) \ge q^T\val \ge q(\var) \cdot \val(\var) \ge \val(\var)$ .
Hence, we get from the Bound Proof Principle that
$\vbound_N(\var) \in O(N^{\varsum(\layer+1)})$ for all variables $\var$ with $\vExp(\var) = \layer+1$.
\qed
\end{proof}

\section{Proof of Lemma~\ref{lem:vars-bounded}}

\begin{proof}
We prove the claim by induction on $\layer$.

For $\layer = 0$, we have $\{ \vExp(\var) \mid \var \in \vars, \vExp(\var) < 0\} = \emptyset$ and $\varsum(\layer) = 1 = 2^0$.
Thus, the claim holds.

We assume the claim holds for some $\layer \ge 0$ and show the claim for some $\layer+1$.
Let $m$ be the number of variables $\var \in \vars$ such that $\vExp(\var) < \layer$.
By induction assumption we have $\varsum(\layer) \le 2^m$.

We now consider a variable $\var$ with $\vExp(\var) = \layer$.
From Lemma~\ref{lem:var-bound-relation} we have that $\vExp(\var) \le \varsum(\layer) \le 2^m$.

There are $(k-m)$ variables $\var$ such that $\vExp(\var) = \layer$.
Hence, we have\\
$\varsum(\layer+1) = \sum_{\var \in \vars, \vExp(\var) < \layer+1}  \vExp(\var) =$\\
$\sum_{\var \in \vars, \vExp(\var) < \layer} \vExp(\var) + \sum_{\var \in \vars, \vExp(\var) = \layer} \vExp(\var) =$\\
$\varsum(\layer) + (k-m) \vExp(\var) \le$\\
$2^m + (k-m)2^m = (k-m+1)2^m \le 2^{k-m}2^m= 2^k$.
\qed
\end{proof}

\section{Proof of Theorem~\ref{thm:exponential}}

\begin{proof}
  By assumption there are no $\transition \in \transitions(\vass)$, $\var \in \vars$ with $\layer < \tExp(\transition) + \vExp(\var) < \nobound$  (*).

  We now will argue that the cycles $\cycle(\node)$ and sets $U = \{ \var \in \vars \mid \vExp(\var) \le \layer\}$, $W = \{ \var \in \vars \mid \vExp(\var) > \layer\}$ satisfy the requirements of Lemma~\ref{lem:exponential}.
  The claim then directly follows from the application of Lemma~\ref{lem:exponential}.

  By Lemma~\ref{lem:constraint-system-property-two} we have
  $\sum_{\node \in \layerF(\layer)} \valueSum(\cycle(\node))(\var) \ge 1$ for variables $\var \in \vars$ with $\vExp(\var) > \layer$.
  It remains to show $\valueSum(\cycle(\node))(\var) \ge 0$ for variables $\var \in \vars$ with $\vExp(\var) \le \layer$ and nodes $\node \in \layerF(\layer)$.

  Let $\node \in \layerF(\layer)$ be a node in layer $\layer$ and let $\var \in \vars$ be a variable with $\vExp(\var) \le  \layer$.
  Let $\node' \in \layerF(\layer - \vExp(\var))$ be the unique ancestor of $\node$ in layer $\layer - \vExp(\var)$.
  We show that $\vassF(\node) = \vassF(\node')$.
  Assume otherwise.
  Then there is a transition $\transition \in \transitions(\vassF(\node')) \setminus \transitions(\vassF(\node))$ with $\layer - \vExp(\var) < \tExp(\transition)$.
  However, this contradicts the assumption (*).
  Hence we get that $\vassF(\node) = \vassF(\node')$.
  Thus, $\node$ is the sole descendant of $\node'$ in layer $\layer$.
  We can then deduce $\valueSum(\cycle(\node))(\var) \ge 0$ from Lemma~\ref{lem:constraint-system-property-one}.
  \qed
\end{proof}

\end{document}